\documentclass[11pt]{article}
\usepackage[applemac]{inputenc}

\usepackage[final]{graphicx}
\usepackage{amssymb,amsmath,enumerate}

 \normalbaselines

\newtheorem{Theorem}{Theorem}[section]
\newtheorem{Definition}[Theorem]{Definition}
\newtheorem{Proposition}[Theorem]{Proposition}
\newtheorem{Lemma}[Theorem]{Lemma}

\newtheorem{Remark}[Theorem]{Remark}

\newtheorem{Assumption}[Theorem]{Assumption}

\def\qed{\hfill\hbox{\hskip 6pt\vrule width6pt height7pt
depth1pt  \hskip1pt}\bigskip}

\newtheorem{proposition}[Theorem]{Proposition}
\usepackage[usenames]{color}

\begin{document}

\title{The Value of Timing Risk}

\author{\begin{tabular}{ccccc}
Jir\^o Akahori\footnote{
The author was supported by JSPS KAKENHI Grant Number $23330109$,
$24340022$, $23654056$, $25285102$, and
the project RARE -318984 (an FP7 Marie Curie IRSES).} & Flavia Barsotti\footnote{The views expressed in this paper are those of the author and should not be attributed to UniCredit or to the author as representative of UniCredit.} &  Yuri Imamura \footnote{
The author was supported by JSPS KAKENHI Grant Number $24840042$.}
\\
{\footnotesize Dept. of Mathematical Sciences
} & {\footnotesize Group Financial Risk Methodologies}  & \footnotesize Dept. of Business Economics \\
{\footnotesize Ritsumeikan University, Japan} & {\footnotesize UniCredit Spa, Italy}  & {\footnotesize Tokyo University of Science, Japan}\\
{\footnotesize akahori@se.ritsumei.ac.jp} &{\footnotesize Flavia.Barsotti@unicredit.eu}  &
\footnotesize imamuray@rs.tus.ac.jp
\end{tabular}
}

\date{}
\maketitle

\textwidth=160 mm \textheight=220mm \parindent=8mm \frenchspacing \vspace{ 3
mm}

\begin{abstract}
The aim of this paper is to provide a mathematical contribution on the semi-static hedge of {\textit{timing risk}} associated to positions in American-style options under a multi-dimensional market model. Barrier options are considered in the paper and semi-static hedges are studied and discussed for a fairly large class of underlying price dynamics. Timing risk is identified with the uncertainty associated to the time at which the payoff payment of the barrier option is due. Starting from the work by \cite{CP}, where the authors show that the timing risk can be hedged via static positions in plain vanilla options, the present paper extends the static hedge formula proposed in \cite{CP} by giving sufficient conditions to decompose a generalized timing risk into an integral of knock-in options in a multi-dimensional market model. A dedicated study of the semi-static hedge is then conducted by defining the corresponding strategy based on positions in barrier options. The proposed methodology allows to construct not only first order hedges but also higher order semi-static hedges, that can be interpreted as asymptotic expansions of the hedging error. The convergence of these higher order semi-static hedges to an exact hedge is shown. An illustration of the main theoretical results is provided for i) a symmetric case, ii) a one dimensional case, where the first order and second order hedging errors are derived in analytic closed form. The materiality of the hedging benefit gain of going from order one to order two by re-iterating the timing risk hedging strategy is discussed through numerical evidences by showing that order two can bring to more than $90\%$ reduction of the hedging 'cost' w.r.t. order one (depending on the specific barrier option characteristics).
\end{abstract}

\section{Introduction}
The uncertainty and high volatility characterizing financial markets make both pricing and hedging activities playing a key role from a risk management perspective. The financial literature treating pricing and hedging problems is very huge: here we concentrate on a specific type of hedging strategy associated to a certain class of exotic derivatives traded on the OTC market. The present paper focuses on static and semi-static hedge of American-style contracts exchanged on the OTC market under a mathematical framework characterized by multi-dimensional diffusion processes.
Barrier options represent some of the most popular path-dependent contracts traded in the financial markets. Let us consider a European up-and-in barrier option written on the underlying $X$, which pays the payoff $F(X_T)$ at maturity $T$ in case the barrier $B$ is hit at maturity $T$. In this case, the event triggering the payment can occur only at maturity, thus the only uncertainty at the beginning of the contract is about the value of the underlying at maturity $T$. Let us now consider the corresponding American-style barrier option, to understand how the timing risk component arises. An American-style barrier option is indeed a path-dependent exotic contract in which the holder of the position does not know the time at which the barrier will be crossed and thus the payoff will be paid. In case of a knock-in option, the payoff will be paid to its holder only in case the barrier $B$ is crossed during the life of the contract, but the time $\tau$ at which this will happen is unknown. And since the time of the payoff payment is unknown, holding a position on this type of derivatives embeds a {\textit{timing risk}} component that must be hedged. This timing risk component exists also in case of digital American-style options: even if the payoff is known (e.g. a rebate can be fixed at inception) the timing risk component is still in.\\

The aim of the paper is to show how to hedge the {\textit{timing risk}} associated to positions in American-style barrier options under a multi-dimensional market model: semi-static hedging strategies are studied and discussed and the corresponding hedging error derived.
\\
Let us consider an agent holding a position on the OTC market on a derivative contract. This position implies an exposure to some market risk factors, and in the case of American-style barrier options can involve optimal exercise rules. In order to hedge the risk associated to this position, the agent can adopt a static hedge strategy by selecting proper derivative instruments to capture and mirror the risk associated to his portfolio. The hedge will be an approximation of the perfect hedge needed to make the position a non-risky one. Starting from the work by \cite{CP}, where the authors show that under Black-Scholes assumptions, 
European-style barrier options can be hedged by a static position of 
two plain vanilla options: one is a long call, and the other is a short put, when the European barrier option to be hedged is a knocked-out call option.
In case the boundary is hit, the portfolio is liquidated at zero cost. Otherwise, it has the same pay-off as the hedged portfolio. 
This kind of strategy, where the position will be changed
at most once,
is often referred to as {\textit{semi-static hedge}}. However, frictions in the market can make the hedge not exactly working in practice: if the market does not believe in Black-Scholes setting, the strategy 
might fail and the two positions (hedged position on knock-out barrier option/hedging position on vanilla options) at the hitting time may have a different price. The cost associated to this event is called {\textit{hedging error}}. \\

Among the numerous contributions on static hedging existing in the financial literature, the work by \cite{BC} is the first example of study showing that, under Black-Scholes setting, static hedges of single barrier options and look-back options can be defined. 
Some extensions of their main results have been provided for i) the case of dividend paying underlying assets (e.g. \cite{CC}), ii) more exotic distribution of the underlying asset price dynamics (e.g. \cite{CEG}), iii) multiple barrier options (e.g. \cite{CC}, \cite{ImTa}). 
In their study the reflection principle and its variant
play a central role, and
the hedge is constructed 
by plain options with the same maturity as the barrier option 
to be hedged. 
In this direction, {the work \cite{CN} provides}
a complete and elegant mathematical solution 
to the problem in the one-dimensional diffusion case. 
Contrasting this {\em strike-spreads}
(c.f. \cite{NP}) approach, 
the work \cite{DEK} 
initiated {\em calendar-spreads}
approach, that is, 
the static hedge 
by plain options with distinct maturities.
Though it shares the idea that 
zero-cost liquidation 
at the barrier, 
rather it stands on 
that the options ``span" the space
of random variables. 
To determine the static position of 
the hedging options, 
an explicit form of option prices 
at the barrier is required. 
In this line, the paper 
\cite{FINK} modified it so as to 
work in the Heston-stochastic volatility setting,
\cite{NP} discussed optimal choice of hedging tool, 
and \cite{ShTaTo}
applied asymptotic expansion 
of option prices proposed 
in \cite{Ta}. 

{
The main results achieved in \cite{CP} can be ``intermediate'':
they are 
summarized as follows: i) the simple timing risk can be decomposed into 
an integral (with respect to maturity) of knock-in options, ii) by applying Bowie-Carr's semi-static hedge formula (see \cite{BC}) of strike-spread approach
to each of them under a Black-Scholes environment, 
the integral representation
implies that 
the semi-static hedge portfolio consists of (infinitesimal amount of)
options with (continuum of) distinct maturities, 
which should be discretized in practice just as 
calendar-spread approach. 



}
The present paper extends the above results
by giving sufficient conditions to decompose a {\textit{generalized timing risk}} into an integral of knock-in options in a multi-dimensional market model by considering a fairly large class of Markovian underlying dynamics (Theorem \ref{GCP}). 
Then, the mathematical contribution of the paper embeds a methodological proposal for the construction of higher order semi-static hedges (Theorem \ref{GSH}), interpreted as asymptotic expansions of the first order hedging error. 
This is done via an iterated procedure based on generalized timing risk identification, semi-static hedging strategy, error computation and re-iteration of the procedure. The convergence of the higher order semi-static hedges to an exact hedge is shown (Theorem  \ref{HSS}). 
Mathematically, our results 
largely rely on 
{\em parametrix},
a classical way to construct 
the fundamental solution 
of partial differential equations
(see e.g. \cite{MR0181836}).
It dates back to \cite{EEL} but
recently some striking 
applications to finance have been reported (see \cite{BK}, \cite{C}, among others). Asymptotic expansion of the price of barrier option
is discussed in \cite{KaTaYa} and \cite{ShTaYa}, but to the best of our knowledge: neither the expansion of the static hedging has been deeply studied, nor an {\em exact} hedge under a general framework has been derived (among papers dealing with a calendar-spread approach).
\\

An illustration of the main theoretical results is provided and discussed for i) a symmetric case, ii) a one dimensional case, where the first order and second order hedging errors are derived in analytic closed form\footnote{{More general multi-dimensional cases
are studied in \cite{FJY2}.}}.
The materiality of the hedging benefit gain of going from order one to order two by re-iterating the timing risk hedging strategy is discussed through numerical evidences by showing that order two can bring to a $90\%$ reduction of the hedging 'cost' w.r.t. order one.\\

%

The paper is organized as follows. Section \ref{TheResults} recalls the main aspects of semi-static hedge and provides the main theoretical results proposed in the paper (Theorem \ref{GCP}, Theorem \ref{GSH}, and Theorem \ref{HSS}) for the multi-dimensional case. Section \ref{secsymm} gathers two applications of the convergence result for the semi-static hedge expansion stated in Theorem \ref{HSS} for i) a symmetric multi-dimensional case; ii) a one dimensional case. Section \ref{simu} derives analytic expressions for the first order and second order hedging errors in the one dimensional case by discussing the materiality of the hedging benefit gain of going from order one to order two supported by numerical evidences. Section \ref{conc} provides some conclusive remarks and ideas for future research. 
The proofs of the main theoretical results of the paper are gathered in a technical appendix (Appendix \ref{tech_app}).

\section{Hedging Timing Risk}\label{TheResults}
This Section discusses the main aspects of semi-static hedge of barrier options and states the main theoretical results proposed in the paper for the multi-dimensional case.

Let us consider an agent holding a portfolio of path-dependent American-style barrier options: this position embeds a timing risk that must be hedged by the investor. This timing risk component is associated with time $\tau$, i.e. the stopping time at which the barrier can be crossed by the underlying asset by triggering the payoff payment event. The uncertainty related to payment time must be monitored by the agent and the risk of the position can be controlled and reduced via a semi-static hedge strategy.

In the present paper we first generalize the Carr-Picron approach \cite{CP} to the multi dimensional case: the proposed methodology provides a sufficient condition under which 
a {\textit{generalized timing risk}} can be decomposed into an integral 
of knock-in options (Theorem \ref{GCP}) in a multi-dimensional diffusion market model. Then, a study of the hedging error is conducted and the construction of an iterated semi-static hedge (Theorem \ref{GSH}) strategy \`a la Carr-Picron \cite{CP} is proposed. We show how to derive the first order hedging error; moreover, the proposed methodology allows to derive higher order hedging errors: the higher order terms can be seen as asymptotic expansion of the hedging error. The convergence of these higher order semi-static hedges to an exact hedge is discussed (Theorem  \ref{HSS}). 

\subsection{A Generalized Carr-Picron Formula for Timing Risk}

Let us assume a multi-dimensional mathematical setting and consider a financial market where $ d $-risky assets and one non-risky asset with constant rate of return $r$ are traded. The prices of the risky assets are driven by a generic strong Markov process $ X $. Let $ \tau $ be a stopping time with respect to the natural filtration of
$ X $ related to the triggering event underlying the American-style barrier option involved in the problem. In order to hedge the timing risk component associated to the position in risky barrier options, we are interested 
to monitor and cover some portion $\psi$ of the price of the barrier option 
{at the stopping time $ \tau $}.\\

Let $ F \in C_b (\mathbf{R}^d ) $ be the pay-off of the option 
and $ T $ be the contract maturity. 
At time $ \tau<T $, the price $p$ of the option is 
\begin{equation}
 p(\tau):=e^{-r (T-\tau)} \mathbb{E} [ F(X_T) | \mathcal{F}_\tau],
 \end{equation}
where the expectation $\mathbb{E}$ is taken under the risk neutral measure, $F$ is the payoff function, $T$ the option maturity, $X_T$ the underlying price at maturity, $r$ the risk free rate of return in the economy.
Then the {\textit{timing risk}} associated to this position can be defined and evaluated at each generic time $t$ as indicated below.

\begin{Definition}\label{deftmgr}
The {\textit{\bf{timing risk}}} at time $t$ associated to the position on a barrier option depends on the stopping time $\tau<T$ and is defined as:
\begin{equation}\label{tmgrsk}
\begin{split}
\mathrm{Tr}_t(F, T, \tau, \psi)&:= e^{-r(T-t) }
\mathbb{E} [1_{ \{ \tau < T \} } \psi (T-\tau) \mathbb{E} [F(X_T) |\mathcal{F}_\tau] | \mathcal{F}_t],
\end{split}
\end{equation}
where $T$ is the option maturity, $X_t$ captures the underlying price dynamics and $F(\cdot)$ is the option payoff function.
We suppose here that the application $\tau \rightarrow \psi(\tau) \in C^1$, with $\psi$ representing the portion of the option price t be covered. 
\end{Definition}
Notice that by re-writing Equation \ref{tmgrsk} under the following equivalent formulation 
\begin{equation}\label{tmgrsk_2}
\begin{split}
\mathrm{Tr}_t(F, T, \tau, \psi)&= 
 \mathbb{E} [ e^{-r(\tau-t)} e^{-r (T-\tau)} 1_{ \{ \tau < T \} } \psi (T-\tau) \mathbb{E} [ F(X_T) | \mathcal{F}_\tau] |\mathcal{F}_t ] 
,
\end{split}
\end{equation}
 it is evident that the payment is due at time $ \tau $ and $ \mathrm{Tr}_t $ is the value at time $ t $ of the payment, i.e. this quantity is the {\textit{timing risk}} associated to the stopping time $\tau$ evaluated at time $t$.\\ 

The paper shows how to build a semi-static hedging strategy of a {\textit{generalized timing risk}} associated to a stopping time, thus highlighting how the timing risk can be decomposed into an infinitesimal amount of continuum of knock-in options. This can be intuitively explained as follows.

{
Let $ 0 \leq s_0 <s_1 < s_2 < \cdots < s_n \leq T $ and 
$ \Pi $ denote the associated partition. 
Suppose that for each $ s_i $, $ i=1, \cdots, n $, 
we have a knock-in option with payoff $ G (s_i, X_{s_i}) $ and amount
$ s_i -s_{i-1} $. Then, assume that for each maturity $ s_i $, 
the payoff of the barrier option is reinvested into the unique non-risky asset available in the market and the knock-in time of the options is $\tau$. 
Thus, the portfolio value at time $ t $ is $PV^\Pi_t$ defined as
\begin{equation*}
\begin{split}
PV^\Pi_t:=& \sum_{i=1}^n e^{-r(s_i-t) } e^{-r(T-s_i)} \mathbb{E}[ 1_{ \{\tau < s_i\}} G (s_i, X_{s_i}) | \mathcal{F}_{t \wedge \tau}]
 ( s_i -s_{i-1}) \\
 =&  e^{-r(T-t)} \sum_{i=1}^n \mathbb{E}[ 1_{\{\tau < s_i\}} G (s_i, X_{s_i}) | \mathcal{F}_{t \wedge \tau}]
 ( s_i -s_{i-1}).
\end{split}
\end{equation*}
We shall consider the limiting object of $ PV^{\Pi} $: 
\begin{equation*}
\begin{split}
& PV_t:= \lim_{|\Pi| \to 0}
PV^\Pi_t
= e^{-r(T-t)} \int_0^T \mathbb{E}[ 1_{\{\tau < s \}} G (s, X_{s}) | \mathcal{F}_{t \wedge \tau}]
ds
\end{split}
\end{equation*}
representing the portfolio value at time $ t $
of infinitesimal amounts of 
continuum of knock-in options.
This idea is used to 
express the hedging portfolio of a timing risk.\\}

{
Assumption \ref{aspGCP} and Lemma \ref{fof2} reported below are useful before stating the first main result (Theorem \ref{GCP}).
\begin{Assumption}\label{aspGCP}Let us consider the following set of assumptions.
\begin{enumerate}
\item The strong Markov process $ X $ has a smooth density; there exist a function 
$ q \in C^{1,2,2} ((0,\infty) \times \mathbf{R}^d \times \mathbf{R}^d) $
such that $ \mathbb{E} [ f (X_t) | X_s ]
= \int_{\mathbf{R}^d} q (t-s, X_s, y) \,f(y) \,dy $ for any bounded measurable function $ f $ on $ \mathbf{R}^d $.
\item There exists a function $ p^F: [0,\infty) \times 
\mathbf{R}^d \times \mathbf{R}^d $ such that $ p^F(\cdot,\cdot,y)$ 
in $C^1 \times C^2$ for all $y \in \mathbf{R}^d$ satisfying
\begin{equation}\label{delta2}
\lim_{s \uparrow t} \int_{\mathbf{R}^d} q (s,x,z) p^F ({t-s},z,y) \,dz 
= q (t, x, y);
\end{equation}
moreover, the quantity 
\begin{equation}\label{def:cf}
c_F:= \psi(0) \int_{\mathbf{R}^d} p^F (t, X_{\tau},y) F(y) \,dy,
 \end{equation}
   is a constant independent of
$t >0 $ almost surely. 
\item 
The function $q (T-s,x,y) ( L_y -\partial_s) p^F (s, y,z) F(z)$ is integrable in $ (s, y, z) \in [0,T] \times \mathbf{R}^d 
\times \mathbf{R}^d  $ for any $ T > 0 $ and $ x \in D $, where $ L_y $ is the infinitesimal generator
of $ X $ acting on the variable $y$.
\item $ L^*_y q (t,x,y)  = \partial_t q $, where $ L^*_y $ denotes 
the adjoint operator (w.r.t. the Lebesgue measure) of $L$, acting on variable $ y $. 
\end{enumerate}
\end{Assumption}


Let us define 
\begin{equation}\label{hpf}
h^{p^F} (s,x,y) :=  \left( \psi' (s) + \psi (s) 
(L_x -\partial_s ) \right) p^F (s, x,y)
\end{equation}
and observe that point 3. in Assumption \ref{aspGCP} allows to state that $ q(t-s, x,y) h^{p^F} (s,y,z) F (z) $ is integrable in $ (s, y, z) \in [0,T] \times \mathbf{R}^d 
\times \mathbf{R}^d  $ for any $ T > 0 $ and $ x \in D $.\\

Let us consider the following Lemma (see e.g. \cite{MR0181836} or \cite{BK} for similar results obtained via parametrix arguments). 
\begin{Lemma}\label{fof2}
For $ t > 0 $ and $ (x,y) \in \mathbf{R}^d \times \mathbf{R}^d $, the following holds:
\begin{equation*}
\psi(t) q (t,x,y) - \psi(0) p^F(t,x,y)
= \int_0^t \int_{\mathbf{R}^d} q (t-s, x,z)
h^{p^F}(s,z,y)
 \,dz ds,
\end{equation*}
where $p^F$ satisfies Equation (\ref{delta2}) and $h^{p_F}$ is defined in Equation (\ref{hpf}).
\end{Lemma}

\begin{proof} See Appendix \ref{prLemma}.
\qed
\end{proof}

 }
\begin{Theorem}\label{GCP}
Under Assumption \ref{aspGCP}, we can state the following result:
\begin{equation}\label{GTMR2}
\begin{split}
&\mathbb{E}[1_{ \{ \tau \leq T \} } \psi (T-\tau)\mathbb{E} [F(X_T) |\mathcal{F}_\tau] | \mathcal{F}_t] 
\\& = c_F P (\tau \leq T| \mathcal{F}_t) + 
\int_0^T 
\mathbb{E} [ 1_{ \{ \tau \leq s\} } \int_{\mathbf{R}^d}
h^{p^F}(T-s,X_s,y)F(y) \,dy
| \mathcal{F}_{\tau \wedge t} ]
ds,
\end{split}
\end{equation}
with $c_F$ defined in Equation (\ref{def:cf}), $p^F$ satisfying Equation (\ref{delta2}) and $h^{p_F}$ defined in Equation (\ref{hpf}).\\
\end{Theorem}

\begin{proof} See Appendix \ref{pr_GCP}.
\qed
\end{proof}


The result in Equation \eqref{GTMR2} extends the static hedge formula proposed in \cite{CP}: Theorem \ref{GCP} gives sufficient conditions to decompose a {\textit{generalized timing risk}} into an integral of knock-in options in a multi-dimensional market model by considering a fairly large class of Markovian underlying dynamics.
Indeed, for $ F \equiv 1 $ and $ \psi (u)= e^{-r(T-u)}$, we can take 
$ p^F = q $ and $ c_F = e^{-rT} $ to satisfy Assumption \ref{aspGCP}. Since in this case
we have $ h^{p^F}(u,x,y) = r e^{-r(T-u)} q (u,x,y) $, Equation \eqref{GTMR2} reduces to:
\begin{equation}\label{simpleex}
\mathbb{E}[ e^{-r \tau} 1_{\{\tau \leq T\}} |\mathcal{F}_t ]
= e^{-rT} P (\tau \leq T | \mathcal{F}_t) +
r\int_0^T \mathbb{E}[ e^{-r s} 1_{\{\tau \leq s \}} |\mathcal{F}_{t \wedge \tau} ]\,ds.
\end{equation}
Letting $ T \to \infty $, we have
\begin{equation*}\label{Carr-Picron}
\mathbb{E}[ e^{-r \tau}  |\mathcal{F}_t ]
= 
\int_0^\infty \mathbb{E}[ e^{-r \tau} 1_{\{\tau \leq s \}} |\mathcal{F}_{t \wedge \tau} ]\,ds,
\end{equation*}
provided that $ P (\tau < \infty) =1 $. 
 Notice that we obtain Equation (\ref{simpleex}) 
directly by integration by parts formula. 
Indeed we see that 
\begin{equation*}
\begin{split}
\mathbb{E}[ e^{-r \tau} 1_{\{\tau \leq T\}} |\mathcal{F}_t ]
&= 
\int_0^T e^{-r s} P(\tau \in ds |\mathcal{F}_t )
\\&= 
e^{-rT}P(\tau < T |\mathcal{F}_t )
-P(\tau < 0 |\mathcal{F}_t )
+r
\int_0^T e^{-r s} P(\tau < s |\mathcal{F}_t )ds
\\&= 
e^{-rT}P(\tau < T |\mathcal{F}_t )
+r
\int_0^T e^{-r s} P(\tau < s |\mathcal{F}_{t \wedge \tau} )ds.
\end{split}
\end{equation*}
The last equation is valid since the set $\{\tau <s\}$ is $\mathcal{F}_{\tau}$-measurable.  

\subsection{Hedging Error in Semi-Static Hedge of Barrier Options as a Timing Risk}\label{subsec:HETR}

The aim of this Subsection is to provide the key intuition concerning the relationship between the semi-static hedging error and the {\textit{value of timing risk}}. Indeed, the first contribution of the present paper is to show that the hedging error of Bowie-Carr type \cite{BC} is a {\textit{timing risk}}. Then, by applying a generalized Carr-Picron \cite{CP} formula we can decompose the timing risk as continuum of knock-in options. Finally, for each knock-in option we apply the semi-static hedge of Bowie-Carr. Leveraging on the intuition linking hedging error and timing risk is the key contribution for the re-iteration of the semi-static hedging strategy of timing risk and the expansion of the first order hedging error into higher order errors with decreasing magnitude.\\



Let us first recall what is a semi-static hedging strategy. 
From now on, $ X $ is a diffusion process and 
$ \tau $ is the first exit time of $ X $ out of a domain $ D \subset \mathbf{R}^d $. We want to hedge 
the knock-out option whose pay-off is $ F(X_T) 1_{ \{ \tau > T\} } $
by holding two plain vanilla options: a long position on an option with payoff $ F(X_T) 1_{\{X_T \in D \}} $, 
and a short position with payoff $-\hat{F}$.
Here, we assume that 
$ \hat{F} = 0 $ on $ D $. Let us notice that: 
\begin{itemize}
\item if $ X $ never exits the domain $ D $, 
then the hedge works exactly since the short position is mirroring the long position; 
\item if the agent liquidates the portfolio at time $ \tau < T $, he has the following cost $C_\tau $:
\begin{equation*}
C_\tau =e^{-r (T-\tau)} E [ ( F(X_T)1_{\{X_T \in D \}} - \hat{F} (X_T) ) | \mathcal{F}_\tau ].
\end{equation*}
\end{itemize}
If $C_\tau =0$, then the static hedge works
perfectly, otherwise the term $C_{\tau}$ can be interpreted as {\textit{hedging error}}.\\

In a similar way, we can consider the
static hedge of a knock-in option with payoff
$ F(X_T) 1_{ \{ \tau < T\} } $ by holding a long positions with payoff $ F (X_T) 1_{ \{X_T \in D^c\} } $ and the option with payoff $ \hat{F} $. In this case, we have: 
\begin{itemize}
\item if $ X $ never exit $ D $, 
then the hedge works exactly and no cash is exchanged (nothing versus nothing). 
\item On 
$ \{ \tau < T \} $, the hedger changes the positions at time $ \tau $  
by selling the option of pay-off $ \hat{F} $
and buying he one with pay-off $ F (X_T) 1_{ \{X_T \in D \}} $. 
The cost of this operation is
\begin{equation}\label{eq:cost_in}
e^{-r (T-\tau)} E [ ( F (X_T) 1_{ \{X_T \in D \}}  - \hat{F} (X_T) ) | \mathcal{F}_\tau ]
= C_\tau.
\end{equation}
\item At maturity $ T $, the pay-off is zero, since the payoff
$ F (X_T) $ exactly compensates $ F (X_T) 1_{ \{ X_T \in D \}} + F (X_T) 1_{ \{ X_T \in D^c \}} $.
\end{itemize}
To generalize and state the first result, let us observe that in both cases described above the {\textit{hedging error}} at time $ t $ associated to the stopping time $\tau$ is
\begin{equation}\label{hdgerr}
\begin{split}
\mathrm{He}^{\tau}_t:=e^{-r (T-t)}E[ E [1_{ \{\tau < T
\} } ( F(X_T) 1_{\{X_T \in D \}} - \hat{F} (X_T) ) | \mathcal{F}_\tau ]| \mathcal{F}_t]. 
\end{split}
\end{equation}

By recalling Definition \ref{deftmgr} for the timing risk and Equation \eqref{tmgrsk} we can easily deduce that the {\textit{hedging error}} in Equation \eqref{hdgerr} can be interpreted as a {\textit{timing risk}} associated to the position in barrier options:
we have that 
$ \mathrm{He^{\tau}_t}= \mathrm{Tr}_t (\pi (F), T, \tau, 1) $, where 
\begin{equation*}
\pi(F) 
=F 1_{ \{x \in D \}}- \hat{F}. 
\end{equation*}
We also define $ \pi^\bot $ by
\begin{equation*}
\pi^{\bot}(F) 
= F(x) 1_{ \{x \in D^c \}} + \hat{F}. 
\end{equation*}
Note that here we assume
$ \hat{F} = \widehat{ F1_{ \{x \in D\}}} $, 
which implies $ \widehat{\pi (F)} = \hat{F} $, $ \pi^2 (F) 
= \pi(F)$, and so on. \\


Let us introduce the following set of assumptions (Assumption \ref{aspGSH}) useful to state the result reported in Theorem \ref{GSH}.


\begin{Assumption}\label{aspGSH}
Starting from the set of assumptions introduced in Assumption \ref{aspGCP}, and by considering function $\pi(F)$, we further 
assume that $ c_{\pi(F)} = 0 $. Thus, we are dealing with the case:
$ \int_{D} p^{\pi(F)} (t, X_{\tau},y) F(y) \,dy = 
\int_{D^c} p^{\pi(F)} (t, X_{\tau},y) \hat{F} (y) \,dy
$ almost surely.
\end{Assumption}

Assumption \ref{aspGSH} requires that the symmetry property we worked in the previous case is now the one w.r.t. a "reflection".

\begin{Theorem}\label{GSH}
Under Assumption \ref{aspGSH}, 
the value of the knock-out/knock-in option 
is decomposed into that of the hedging portfolio and 
an infinitesimal amount of 
continuum of knock-in options. More precisely, 
\begin{equation}\label{KOGSH}
\begin{split}
&\mathbb{E} [ F (X_T) 1_{\{ \tau \geq T \} }|\mathcal{F}_{\tau \wedge t} ]
\\& =\mathbb{E}[ \pi(F)(X_T) |\mathcal{F}_{\tau \wedge t}  ] 
- \int_0^T
\mathbb{E} [ 1_{ \{ \tau \leq s \}} \int_{\mathbf{R}^d}
h^{p^{\pi(F)}}(T-s,X_s,y)\pi (F)(y) \,dy 
| \mathcal{F}_{\tau \wedge t} ] \,ds,
\end{split}
\end{equation}
and 
\begin{equation}\label{KIGSH}
\begin{split}
&\mathbb{E}[ F (X_T) 1_{\{ \tau \leq T \} }|\mathcal{F}_{\tau \wedge t}  ]
\\& =\mathbb{E} [ \pi^\bot (F)(X_T) |\mathcal{F}_{\tau \wedge t}  ] 
+ \int_0^T
\mathbb{E} [ 1_{ \{ \tau \leq s \}} \int_{\mathbf{R}^d}
h^{p^{\pi(F)}}(T-s,X_s,y)\pi (F)(y) \,dy| \mathcal{F}_{\tau \wedge t} ] \,ds.
\end{split}
\end{equation}
\end{Theorem}
\begin{proof} See Appendix \ref{pfGSH}.
\qed
\end{proof}


Theorem \ref{GSH} is the base to construct higher order semi-static hedges, that can be interpreted as asymptotic expansions of the first order hedging error. Indeed, the result states that the generalized timing risk can be decomposed into a continuum of knock-in options (extending \cite{CP}). Then, for each knock-in option in which the timing risk has been decomposed, by applying Bowie-Carr \cite{BC} strategy we can further reduce the hedging error associated to the position in barrier options as shown in the following Subsection.

\subsection{Higher Order Semi-Static Hedge}\label{sec:HO}

The contribution of the present paper is to define a {\textit{generalized timing risk}} and show how to compute first order and higher orders hedging errors (with reducing magnitude) associated to semi-static hedges of American-style barrier options. This subsection introduces and states the existence of higher order semi-static hedging strategies via the proposed iteration procedure and discusses the conditions to have the error converging to zero.\\

Theorem \ref{GSH} extends the static hedge formula proposed in \cite{CP} by giving sufficient conditions to decompose a {\textit{generalized timing risk}} into an integral of knock-in options in a multi-dimensional market model by considering a fairly large class of Markovian underlying dynamics. Then, the mathematical contribution of the paper embeds a methodological proposal for the construction of higher order semi-static hedges, interpreted as asymptotic expansions of the first order hedging error. This is done via an iterated procedure based on generalized timing risk identification, semi-static hedging strategy, error computation and re-iteration of the procedure. The convergence of the higher order semi-static hedges to an exact hedge is shown.\\

The proposed methodology to construct the higher order semi-static hedging errors leverages on the following steps:
\begin{enumerate}
\item identify the {\textit{generalized timing risk}} and apply a Carr-Picron \cite{CP} strategy, by obtaining a knock-in options decomposition; 
\item for each knock-in option in which the timing risk has been decomposed, apply Bowie-Carr \cite{BC} strategy;
\item observe that the hedging error associated to a Bowie-Carr \cite{BC} strategy is a {\textit{generalized timing risk}};
\item re-iterate from step 1.
\end{enumerate}

The first key contribution of the paper is to show that the hedging error  la \cite{BC} is a {\textit{timing risk}}; starting from this observation, a generalized \cite{CP} formula is applied to decompose it into a knock-in options representation. Then, for each knock-in option we apply a semi-static hedge (\cite{BC} type strategy) and define a second order error, since we are working with general diffusion with no symmetry. The idea is then to iterate the approach (generalized w.r.t. Carr Picron \cite{CP}) and obtain also higher order semi-static hedges and the corresponding errors. Let us consider hedging error of order $n$: this can be obtained by identifying, at each step of the iteration procedure until $n-1$, the set of knock-in options in which the hedging error has been decomposed. 
Conditions to have the error converging to zero are derived.\\

Theorem \ref{GSH} provides the result for the {\em first order hedging error} based on the semi-static hedge decompositions in terms of knock-in options. 
Let us observe that the integrand of the second term of the right-hand-side of \eqref{KOGSH} and \eqref{KIGSH}
are again a pay-off of knock-in options,
thus the formula can be iterated.
The {\em second order hedge} is
integration of infinitesimal amount $ ds $ of
\begin{equation*}
\pi^\bot \int_{\mathbf{R}^d}
h^{p^{\pi(F)}}(T-s,X_s,y)\pi (F)(y) \,dy 
\end{equation*}
for each $ s $. 
The {\em second order error} associated to the stopping time  
$ \tau $ 
is then
\begin{equation*}
\int_\tau^T \pi \int_{\mathbf{R}^d}
h^{p^{\pi(F)}}(T-\tau,X_\tau,y)\pi (F)(y) \,dy ds.
\end{equation*}
By applying Theorem \ref{GSH}
for each $ s $, 
the {\em second order hedging error}, evaluated at $ t $, can be represented as 
\begin{equation}\label{2err}
\int_0^s
\mathbb{E} [ 1_{ \{ \tau \leq u \}} \int_{\mathbf{R}^d}  \int_{\mathbf{R}^d} h^{p^{\pi(F)}}(s-u,X_u,y_1)
\pi h^{p^{\pi(F)}}(T-s,y_1,y_2)\pi (F)(y_2) \,dy_1 dy_2 | \mathcal{F}_{\tau \wedge t} ] \,du.
\end{equation}
The total error is 
obtained by integrating \eqref{2err} with respect to $s $.\\

By re-iterating the procedure, an asymptotic expansion of the hedging error can be obtained under suitable conditions. Note that we shall work in a setting independent w.r.t. the payoff function $ F $. In this section, $\pi$ means 
the correspondence $ F \mapsto 
F 1_{\{x \in D\}} - \hat{F} $ and 
$ \pi^\bot $ means $ F \mapsto F 1_{\{x \in D^c\}} +\hat{F} $.  
The setting is independent w.r.t. $F$ but dependent on function $ \pi $
or $ \pi^\bot $. \\

{In addition to Assumption \ref{aspGSH}, 
we assume that $ p^{\pi(F)} \equiv p^{\pi} $ 
can be chosen to be independent of $ F $;
it is only dependent on $ q $, $ D $ and $ \pi 
= 1 - \pi^{\bot} $, where $ \pi^{\bot} 
: C_b (D) \to C_b (D^c) $.\\

Let us introduce Assumption \ref{aspHSS} and Assumption \ref{aspHSSc}  useful to state the result given in Theorem \ref{HSS}.

\begin{Assumption}\label{aspHSS}Let us consider the following set of assumptions.
\begin{enumerate}
\item The diffusion process $ X $ has a smooth density; there exist a function 
$ q \in C^{1,2,2} ((0,\infty) \times \mathbf{R}^d \times \mathbf{R}^d) $
such that $ \mathbb{E}[ f (X_t) | X_s ]
= \int_{\mathbf{R}^d} q (t-s, X_s, y) \,f(y) \,dy $ for any bounded measurable function $ f $ on $ \mathbf{R}^d $.
\item There exists a function $ p^\pi: [0,\infty) \times 
\mathbf{R}^d \times \mathbf{R}^d $ such that $ p^F(\cdot,\cdot,y)$ 
in $C^1 \times C^2_b $ for all $y \in \mathbf{R}^d$, 
\begin{equation*}
\lim_{s \uparrow t} \int_{\mathbf{R}^d} q (s,x,z) p^\pi ({t-s},z,y) \,dz 
= q (t, x, y), 
\end{equation*}
and that for any $ F \in C_b (D) $, 
\begin{equation*}
\int_{D} p^\pi (t, X_{\tau},y) F(y) \,dy =
\int_{D^c} p^\pi (t, X_{\tau}, y) (\pi^\bot F) (y) \,dy.
\end{equation*}
\item $ L^*_y q (t,x,y)  = \partial_t q $, where $ L^*_y $ denotes 
the adjoint operator (with respect to Lebesgue measure) of $L$, acting on 
the variable $ y $. 
\item  
$ I_{0=s_0 < s_1 < \cdots < s_N < T}q (T-s_N, x, y_N) \prod_{j=1}^{N} h^{p^\pi} (s_j-s_{j-1}, y_{j}, y_{j-1}) $ 
is integrable in $ (s_1,\cdots, s_N, y_1, \cdots, y_N ) $ 
on $ [0,T]^N \times \mathbf{R}^{d N} $ for $ x \in D $, for arbitrary $ N \geq 1 $, where
$ h^{p^\pi} (s,y,z) 
= ( L_y -\partial_s) p^\pi (s, y,z) $. 
\end{enumerate}
\end{Assumption}

With Assumption \ref{aspHSS}, we can define 
family of operators $ (\mathcal{S}^{p^F} )^N_t $, $ t \in [0,T] $ and
$ N=1,2, \cdots $, inductively by for $f \in C_b (\mathbf{R}^d)$ and $x \in \mathbf{R}^d$.
\begin{enumerate}
\item We define 
\begin{equation}\label{Sp}
  (\mathcal{S}^{p^\pi}  )^1_t f(x) = 
\int_{\mathbf{R}^d} 
h^{p^{\pi}} (t,x,y) \pi(f) (y) \, dy\end{equation}.
\item For $ N \geq 2 $, we thus have:
\begin{equation*}
\begin{split}
(\mathcal{S}^{p^\pi} )^N_t f(x) = 
\int_0^t 
(\mathcal{S}^{p^\pi}  )^1_s
(\mathcal{S}^{p^\pi} )^{N-1}_{t-s} f(x)\,ds.
\end{split}
\end{equation*}
\end{enumerate} 

\begin{Assumption}\label{aspHSSc}
In addition to Assumption \ref{aspHSS}, we assume
that 
\begin{equation*}
\begin{split}
& H_N (x,t) \\
&:= \int_0^{t} \int_0^{s_{N}}\int_0^{s_{N-1}} \cdots \int_0^{s_2} 
\int_{\mathbf{R}^{dN}} \prod_{j=1}^{N} q (t-s_N, x, y_N) |h^{p^\pi} (s_{j} -s_{j-1}, y_{j}, y_{j-1})| 
d  ds_1 \cdots d s_{N} \\
& \hspace{4cm} \times |F(y_0)| dy_0 dy_1 \cdots d y_{N-1}  dy_N
\end{split}
\end{equation*}
where $ 0=s_0 $ and $ x = y_N $ 
converges to zero as $ N \to \infty $ uniformly in $ x $. 
\end{Assumption}
}


We can now state the following theoretical result based on static hedge of timing risk useful to show the convergence of the proposed re-iterated procedure.\\


\begin{Theorem}\label{HSS}
Under Assumption \ref{aspHSS}, for any $ N \geq 1 $ we have:
\begin{equation}\label{HSSform}
\begin{split}
& \mathbb{E} [ F (X_T) 1_{\{ \tau \geq T \} }|\mathcal{F}_{\tau \wedge t}  ]
\,\, (\text{resp.} E [ F (X_T) 1_{\{ \tau \leq T \} }|\mathcal{F}_{\tau \wedge t} ]) \\ 
&= \mathbb{E}[ \pi(F)(X_T) |\mathcal{F}_t ] 
\,\, 
(\text{resp.} E [ \pi^{\bot} F (X_T) 1_{\{ \tau \leq T \} }|\mathcal{F}_t ]) 
\\
& \quad \mp \sum_{k=1}^{N-1} \int_0^T
\mathbb{E} [ \pi^\bot ((\mathcal{S}^{p^{\pi}})^k_{T-s} 
(F) ) (X_s)| \mathcal{F}_{\tau \wedge t} ] \,ds \\
& \qquad \mp  \int_0^T
\mathbb{E} [  1_{ \{ \tau \leq s \}} ( (\mathcal{S}^{p^{\pi}})^N_{T-s} 
(F) ) (X_s)| \mathcal{F}_{\tau \wedge t} ] \,ds,
\end{split}
\end{equation}
where we understand $ \sum_{k=1}^{0} (\cdots) =0 $.
Furthermore, under Assumption \ref{aspHSSc}, we have that
i) $ \sum_{k=1}^N \pi^\bot (\mathcal{S}^{p^{\pi}})^k_{T-s} (x) $
converges uniformly in $ x $, ii)
$ \sum_{k=1}^\infty \pi^\bot (\mathcal{S}^{p^{\pi}})^k_{T-s} (X_s) $
is integrable in $ (s,\omega) $, iii) and the following holds
\begin{equation}\label{exact}
\begin{split}
& \mathbb{E} [ F (X_T) 1_{\{ \tau \geq T \} }|\mathcal{F}_{\tau \wedge t}  ]
\,\, (\text{resp.} \mathbb{E} [ F (X_T) 1_{\{ \tau \leq T \} }|\mathcal{F}_t ]) \\ 
&= \mathbb{E} [ \pi(F)(X_T) |\mathcal{F}_{\tau \wedge t}  ] 
\,\, 
(\text{resp.} \mathbb{E} [ \pi^{\bot} F (X_T) 1_{\{ \tau \leq T \} }|\mathcal{F}_t ]) 
\\
& \qquad \mp  \int_0^T
\mathbb{E} [ \sum_{k=1}^\infty \pi^\bot ( (\mathcal{S}^{p^{\pi}})^k_{T-s} 
(F) ) (X_s)| \mathcal{F}_{\tau \wedge t} ] \,ds,\\
\end{split}
\end{equation}
with $ \mathcal{S}^{p^\pi} $ defined in Equation (\ref{Sp}).
\end{Theorem}

\begin{proof}See Appendix \ref{pfHSS}.
\qed
\end{proof}

Equation \eqref{HSSform}
states that the $ N $-th order hedge is given by
\begin{equation*}
\int_0^T
\mathbb{E} [ \pi^\bot ((\mathcal{S}^{p^{\pi}})^k_{T-s} 
(F) ) (X_s)| \mathcal{F}_{\tau \wedge t} ] \,ds 
\end{equation*}
and the error with first to $ N $-th order hedge is 
\begin{equation*}
\int_0^T
\mathbb{E} [  1_{ \{ \tau \leq s \}} ( (\mathcal{S}^{p^{\pi}})^N_{T-s} 
(F) ) (X_s)| \mathcal{F}_{\tau \wedge t} ] \,ds, 
\end{equation*}
and 
Equation \eqref{exact}
claims that, not only 
the error converges to zero, but
\begin{equation*}
\int_0^T
\mathbb{E} [ \sum_{k=1}^\infty \pi^\bot ( (\mathcal{S}^{p^{\pi}})^k_{T-s} 
(F) ) (X_s)| \mathcal{F}_{\tau \wedge t} ] \,ds
\end{equation*}
hedges the barrier option without error.\\

The proposed methodology is based on the key step of  timing risk identification; then, semi-static hedging strategy application via knock-in options decomposition and Bowie-Carr \cite{CP} strategy application to each knock-in option allow to construct not only first order hedging error but also higher orders errors. The convergence of the higher order semi-static hedges to an exact hedge is then shown via the result stated in Theorem \ref{HSS}.


\section{Applications}\label{secsymm}

This Section provides an illustration of the main theoretical results stated in the previous Section under two special settings. The application of the mathematical results is presented for i) a symmetric case, ii) one dimensional case.  
{
More general multi-dimensional cases
are studied in \cite{FJY2}.
}

\subsection{Symmetric Case}
This subsection considers the case of a diffusion that is {\textit{symmetric}} under the reflection with respect to a hyperplane:

\begin{equation*}
D := \{ x \in \mathbf{R}^d | \langle x, \gamma \rangle > k \},
\end{equation*}
where $ \gamma \in \mathbf{R}^d $ with $ |\gamma| =1 $, 
and $ k \in \mathbf{R} $; in this case, $ \hat{F} $ can be constructed by
the associated reflection as
\begin{equation*}
\theta (x) = x - 2 {\langle \gamma, x \rangle \gamma}
+ 2 
{k \gamma} 
= \left(I- 2 
{\gamma \otimes \gamma} 
\right) x 
+ 2 
{k \gamma}. 
\end{equation*}

Let $ X $ be a diffusion on $ \mathbf{R}^d $, the infinitesimal generator of
which is given by
\begin{equation}\label{genrt}
\frac{1}{2} A(x) \cdot \nabla^{\otimes 2} + b(x) \cdot \nabla
\equiv \frac{1}{2} \sum_{i,j} a_{i,j} (x) \frac{\partial^2}{\partial x_i \partial x_j}
+ \sum_i b_i (x) \frac{\partial}{\partial x_i}
\end{equation}
where the diffusion matrix $ A \equiv (a_{i,j}) $ 
and the drift vector $ b \equiv (b_i)$ satisfy 
Assumption \ref{aspdiff}.

{
\begin{Assumption}\label{aspdiff}
There exist two positive constants, namely $ m $ and $ M $, such that 
\begin{equation*}
m |y|^2 \leq 
\langle A (x) y,y \rangle 
\leq M |y|^2  \quad \forall x,y \in \mathbf{R}^d 
\end{equation*}
with $ A \equiv (a_{i,j}) $  being the diffusion matrix 
and $ b \equiv (b_i)$ the drift vector; $ a_{ij}, b_j  $ have any order of derivatives, 
which are all bounded above.


\end{Assumption}

We note that $A$ and $b$ are Lipschitz continuous under the assumption.
In particular, if we put
\begin{equation*}
a_{\infty} 
:= \left\{ \sum_{i,j} d \max_k \left( \sup_{x \in \mathbf{R}^d} 
| \partial_k a_{i,j} (x) | \right)^2 \right\}^{\frac{1}{2}},
\end{equation*}
we have that
\begin{equation*}
\Vert A (x) - A(y) \Vert \leq a_{\infty} | x -y |.
\end{equation*}

\begin{Remark}
{\em 
Remark that Assumption \ref{aspdiff} implies what follows (see e.g. \cite[Theorem 1.11, Theorem 1.15]{MR0181836}). 
The transition density of $ X $ is given by
\begin{equation*}
q(t,x,y) = P (X_t \in dy |X_0=x)/dy.
\end{equation*}
This transition density exists and has the following properties: i) it is twice continuously differentiable in $ (x,y) $, ii) it is continuously differentiable in $ t $, ii) for any 
$M_0 > M $ and some constant $C_q >0$, it satisfies:
\begin{equation*}
q(t,x,y) \leq  C_q t^{-\frac{d}{2}}\exp\{-\frac{ |x-y|^2}{4M_0t}\},
\end{equation*}
\begin{equation*}
|\nabla q(t,x,y) | \leq  C_q 
t^{-\frac{d+1}{2}}\exp\{-\frac{ |x-y|^2}{4M_0t}\},
\end{equation*}
and 
\begin{equation}\label{adj0}
\partial_t q(t, x,y) = (L_x q)(t,x,y) = (L^*_y q)(t, x,y) 
\end{equation}
where $ L_x $ is the infinitesimal generator of $ X $ (see Equation (\ref{genrt}))
acting on variable $ x $, and 
$ L^*_y $ is the adjoint of $ L $, acting on the variable $ y $
\begin{equation*}
\begin{split}
L^*_y &= \frac{1}{2} \nabla^{\otimes 2}_y \cdot A(y) -  \nabla_y \cdot b (y) \\
& \equiv \frac{1}{2} \sum_{i,j} a_{i,j} (y) \frac{\partial^2}{\partial y_i \partial y_j} 
+ \sum_i \left( \sum_j \frac{\partial  a_{ij}}{\partial y_j}  (y)
- b_i (y) \right) \frac{\partial}{\partial y_i}
+ \frac{1}{2} 
\sum_{i,j} \frac{\partial^2 a_{ij}}{\partial y_i \partial y_j}(y)
- \sum_i \frac{\partial b_i }{\partial y_i} (y). 
\end{split}
\end{equation*}
Moreover, we have
\begin{equation*}
\int_{\mathbf{R}^d} (L^*_y q)(s, x,y) g(y) \,dy = \int_{\mathbf{R}^d} q(s, x,y) L_y g(y) \,dy
\end{equation*}
for any function $ g \in C_0^\infty (\mathbf{R}^d) $ 
(see e.g. \cite{MR0181836}). 
}
\end{Remark}

}
\begin{Proposition}\label{symm}
Suppose that $ A $ and $ b $ are symmetric
under the reflection $ \theta $; 
\begin{equation}\label{PPCS}
A (x) = \Psi A (\theta(x)) \Psi, \quad b(x) = \Psi b (\theta(x)).
\end{equation} 
where $ \Psi = I- 2 {\gamma \otimes \gamma} $, 
which is both orthogonal and symmetric.
Define $ \pi $ by $ \pi F (x) := F(x) - F( \theta (x)) $. 
Then we have that Assumption \ref{aspHSSc} is satisfied 
by taking $p^{\pi} = q$.  
In particular, 
\begin{equation}\label{XPCS}
q(t,x,y) - q(t,x, \theta(y) ) = 0, \quad x \in \partial D, y \in D,
\end{equation}
where $ q(t,x,y) :=P^x (X_t \in dy)/dy  $.
\end{Proposition}

\begin{proof}See Appendix \ref{pfsymm}.
\qed
\end{proof}

Let us look at the symmetry of the diffusion more closely. 
The formula (\ref{XPCS})
can be understood as a consequence that
$ q $ is approximated by  
\begin{equation}\label{DPCS}
p^{A,b}_t (x,y) = (2 \pi)^{-\frac{d}{2}} \{\det A(y)t\}^{-\frac{1}{2}} 
e^{- \frac{1}{2t} \langle A(y)^{-1} (x -y-b(y) t), x -y-b(y) t \rangle},
\end{equation}
through, an Euler-Maruyama approximation of the corresponding stochastic differential equation.
Observe that $ p^{A, b} $ satisfies
\begin{equation*}
p^{A,b}_t (x,y) - p^{A,b}_t (x, \theta(y) ) = 0, 
\quad x \in \partial D, y \in \mathbf{R}^d 
\end{equation*}
is a direct consequence of (\ref{PPCS}). 
In fact, since $ \Psi^2 = I $ and $ x = \theta (x) $ for $ x \in \partial D $,
\begin{equation*}
\begin{split}
& p^{A,b}(t,x, \theta(y) ) \\
&= 
(2 \pi)^{-\frac{d}{2}} \{\det A(\theta(y))t\}^{-\frac{1}{2}} 
e^{- \frac{1}{2t} \langle A(\theta(y))^{-1} (x -\theta(y)-b(\theta(y)) t), x -\theta(y)-b (\theta(y)) t \rangle} \\
&= (2 \pi)^{-\frac{d}{2}} \{\det \Psi A(y) \Psi t \}^{-\frac{1}{2}} 
e^{- \frac{1}{2t} \langle A(y)^{-1} \Psi (\theta(x) -\theta(y)
- \Psi b (y) t), \Psi( \theta(x) -\theta(y)- \Psi b(y) t)  \rangle} \\
& =(2 \pi)^{-\frac{d}{2}} \{\det A(y) t \}^{-\frac{1}{2}} 
e^{- \frac{1}{2t} \langle A(y)^{-1} \Psi^2 (x -y
- b (y) t), \Psi^2( x -y- b(y) t)  \rangle} \\
&= p^{A,b}(t,x,y).  
\end{split}
\end{equation*}


\begin{Remark}[Asymmetry of Stochastic Volatility Models]
{\em
Let us consider the case where 
$ \gamma = (1,0,\cdots, 0) $; the knock-out condition 
is only dependent on the first variable
since $ D = \{ x_1 > k \} $. We can interpret the first variable to be the price process of the 
underlying, and the second to be its volatility, and so on. Note that, by assuming the symmetry 
(\ref{PPCS}) on $ A $, we can write 
\begin{equation*}
A (x) - A(\theta(x) ) = A (x) - \Psi A (x) \Psi 
= \Psi ( \Psi A (x) - A (x) \Psi)
= \Psi [\gamma \otimes \gamma, A(x) ]  
\end{equation*}
for $ x \not\in \partial D $, 
and therefore
\begin{equation*}
| A (x) - A(\theta(x) ) | = 2 \sum_{j=1}^d (a_{1,j} (x))^2. 
\end{equation*}
Based on this result, $ A $ is not continuous at $ \partial D $
unless ``volatility of volatility" $ a_{1,2} (x) $ is zero at
the hyperplane. The continuity of the volatility matrix 
and the non-vanishing of the correlation terms are not compatible.
If we stick to both requirement, we need to abandon 
the symmetry (\ref{PPCS}).}
\end{Remark}

\if8
\subsection{Symmetrization}
In this section, we consider the cases where
$ X $ is a generic diffusion process whose infinitesimal generator 
is (\ref{genrt}) with Assumption \ref{aspdiff}, 
but the coefficients $ A $ and  $ b $ 
do not necessarily satisfy (\ref{PPCS}). 

Let  
\begin{equation*}
\tilde{A} (x) = 
\begin{cases}
A(x) & x \in D \\
\Psi A (\theta(x)) \Psi & x  \not\in D,
\end{cases}
\end{equation*}
the symmetrization of $ A $ with respect to $D$ introduced in \cite{AI}. 

Let 
$ p^{\tilde{A},0} $ be the density 
in terms of the notation defined in (\ref{DPCS}) 
 and 
$\pi F(x) = F(x) - F(\theta(x))$.
Our main result in this section is the following 
\begin{Theorem}\label{symmrz}
We assume tha $ A $ and $ \Psi $ is commutative.
Then Assumption 
\ref{aspHSSc} is fulfilled by taking $ p^\pi = p^{\tilde{A},0} $.
\end{Theorem}
For a proof, see section \ref{pfsymmrz}
\fi

\subsection{One Dimensional Case}\label{subs:onedim}
This Subsection analyzes and discusses the application of the proposed methodology to the one dimensional case, where first and second order hedging errors can be derived in analytic closed form.\\ 

Let us consider the (time homogeneous) one dimensional case:
\begin{equation}
d X_t = \sigma (X_t) \, dW_t + \mu (X_t) \, dt,
\end{equation}
where $ \sigma > 0 $ and 
$ \mu $ are smooth functions with linear growth, 
while 
$ D = [K, \infty) $, with $ X_0 > K $. 

The Lamperti transform:
\begin{equation*}
    s(x) = \int^x \frac{1}{\sigma (y)} dy
\end{equation*}
reduces the problem to the one for  
Brownian motion with drift:
\begin{equation*}
Y_t = s (X_t) = s (X_0) 
+ W_t  +\int_0^t \left( \frac{\mu (s^{-1}(Y_s))}{\sigma (s^{-1}(Y_s))}-\frac{1}{2}\sigma' (s^{-1} (Y_s)) \right) \,ds,
\end{equation*}
with $ F_Y := F\circ s^{-1} $ replacing $ F $ and 
\begin{equation*}
    \tau = \inf \{ s >0 : X_s < K \} =
    \inf \{ s > 0 : Y_s < s (K) \}.
\end{equation*}

\begin{Theorem}\label{1dim}
We assume that $ \sigma > 0 $ is $C^1_{b}$ and $||\mu / \sigma ||_{\infty} < \infty$. 
Let $ \pi $ be such that 
$ \pi^\bot f (x) = f (2 s(K) -x) $ for $ f \in C_b ([s(K), \infty)) $. Then, 
\begin{equation*}
    p^\pi(t,x,y) = \frac{1}{\sqrt{2\pi t} } \exp \left( 
    - \frac{1}{2t} (x-y)^2 \right),
\end{equation*}
with 
\begin{equation}\label{hppi0}
    h^{p^\pi} (t,x,y) = 
	\left\{- \frac{\mu (s^{-1}(x))}{\sigma (s^{-1}(x))}+
	\frac{1}{2} \sigma' (s^{-1} (x)) \right\}\frac{(x-y)}{t} p^{\pi}(t,x,y) 
\end{equation}
satisfies Assumption \ref{aspHSSc}. Consequently, the exact expansion \eqref{exact} holds. 
\end{Theorem}

\begin{proof}See Appendix \ref{pf1dim}.
\qed
\end{proof}

\begin{Remark}
{\em
In $ d \geq 2 $-dimensional case with 
\begin{equation*}
    d X_t = \sum_{j=0}^d V_j (X ) \circ dW^j_t,
\end{equation*}
with the convention $ dW^0 = dt $, 
if the smooth vector fields $ V_1, \cdots V_d $ commute each other,
then it is ``diffeomorphic" to standard Brownian motion 
modulo drift and a straightforward extension of Theorem \ref{1dim} holds.
Most stochastic volatility models are instead
``diffeomorphic" to hyperbolic Brownian motion modulo drift\footnote{For a related discussion, see \cite{II}.},
as pointed out in \cite{HL}. Hyperbolic cases need a different function
$ \pi $, which is not easy to identify. 
}
\end{Remark}

\section{Hedging Error}\label{simu}
To illustrate an application of the proposed framework, analytic results in the one dimensional case are derived for both i) the first order hedging error, ii) the second order hedging error, obtained by re-iterating the timing risk identification and hedging (semi-static hedge decomposition and representation in terms of knock-in options). Let us briefly recall the mathematical setting described in Subsection \ref{subs:onedim} and then state the main results in Proposition \ref{p:FOHE} and Proposition \ref{p:SOHE} below. A comparison between the first and second order hedging errors is provided, by showing the materiality of the hedging benefit gain (i.e. in terms of error reduction in absolute value) derived by passing from order one to order two and re-iterating the hedging procedure.\\

Assumption \ref{oneDim} details the mathematical setting characterizing the one dimensional case.

\begin{Assumption}\label{oneDim}
The dynamics of the underlying process $X_t$ is described by the following SDE:
\begin{equation}
dX_t=\sigma dW_t + \mu dt,
\end{equation}
where $\mu = r-\frac{1}{2}\sigma^2$, with $r, \sigma>0$ and $\mu:=r-\frac{1}{2}\sigma^2$; its solution is thus given by:
$$
X_t=X_0+\sigma W_t + \mu t.
$$
Notice that $X_t$ is a random variable with distribution

$$X_t\sim N(X_0+\mu t,\sigma^2 t).$$

{Let us consider a {\textit{knock-in}} barrier option with underlying $X_t$ and the following option-related quantities:
\begin{itemize}
\item $\log K$ the barrier for the knock-in condition,
\item $K{'}$ the strike price of the {\textit{hedging}} options,
\item $\tau$ the hitting time related to the {\textit{timing risk}}.
\end{itemize}
Then, let the payoff function be
\begin{equation}\label{payoff1}
F(x)=(e^{x}-K')^+,
\end{equation}
where $ 0 < K \leq K' $, and consider function $\theta(\cdot)$ defined as
$$\theta(x) :=2 {\log K}-x.$$
Thus, by substitution, we can write:
\begin{equation}\label{payoff2}
F(\theta(x))=(e^{2{\log K}-x}-K')^+. 
\end{equation}}
\end{Assumption}

The analytic results for both the first and second order hedging errors are reported and discussed in the following subsections.
Results are obtained by leveraging on: identification of {\textit{generalized timing risk}}; {\textit{hedging error}} definition (Equation \eqref{hdgerr}) and equivalence between hedging error and timing risk; semi-static hedge decomposition in terms of infinitesimal amount of knock-in options; semi-static hedge strategy applied to each knock-in option.

\subsection{The first order hedge and its error}
Let us recall Equation \eqref{hdgerr} for the {\textit{hedging error}} definition:

$$\mathrm{He}^{\tau}_t:=e^{-r (T-t)}E[ E [1_{ \{\tau < T
\} } ( F(X_T) 1_{\{X_T \in D \}} - \hat{F} (X_T) ) | \mathcal{F}_\tau ]| \mathcal{F}_t]
$$
and then consider the portfolio characteristics in order to identify the first order hedging error.
Note that by substituting Equations \eqref{payoff1}-\eqref{payoff2} into the hedging error definition, we obtain an integral representation for the first order hedging error in the one dimensional case.
Indeed, under Assumption \ref{oneDim}, the first order 
hedging error associated to the stopping time $\tau$ can be written as follows:
\begin{equation}\label{he_diff}
\mathrm{He}_{\tau}^{(1)}:=
\mathbb{E}[( e^{X_{T -\tau }} -K')^+|X_0=\log K ]
-\mathbb{E}[( e^{2 \log K -X_{T -\tau }} -K')^{+}|X_{0}=\log K  ]
\end{equation}
and can be interpreted as difference between two option prices.\\

The first order hedging error can be characterized in analytic closed form as stated in Proposition \ref{p:FOHE}.

\begin{proposition}\label{p:FOHE} Under Assumption \ref{oneDim}, the {\textit{first order hedging error}} at time $t$ associated to the stopping time $\tau$ is given in closed form as:
\begin{equation}\label{fohe}
\begin{split}
\mathrm{He}_{\tau}^{(1)}&
= K' \left[\mathcal{N}\left(d_1\right)-\mathcal{N}\left(d_2\right)\right]\\
& + Ke^{\frac{\sigma^2(T -\tau )}{2} } \left[
e^{\mu (T -\tau ) }
\mathcal{N}\left(d_2{+} \sigma\sqrt{T -\tau}\right)-e^{-\mu (T -\tau ) }\mathcal{N} 
\left(d_1{+} \sigma\sqrt{T -\tau }\right)\right]
\\
\end{split}
\end{equation}
with
\begin{equation*}
d_1=\frac{\log {\frac{K}{K'}-} \mu (T -\tau )}{\sqrt{\sigma^2 (T -\tau )}}, \quad d_2=\frac{\log {\frac{K}{K'}+} \mu (T -\tau )}{\sqrt{\sigma^2 (T -\tau )}},
\end{equation*}
and $\mathcal{N}(\cdot)$ indicating the cumulative distribution function of a standard normal random variable.


\end{proposition}

\begin{proof}See Appendix \ref{pfFOHE}.
\qed
\end{proof}

\begin{figure}[ht!]
		\centering
		\includegraphics[scale=0.7]{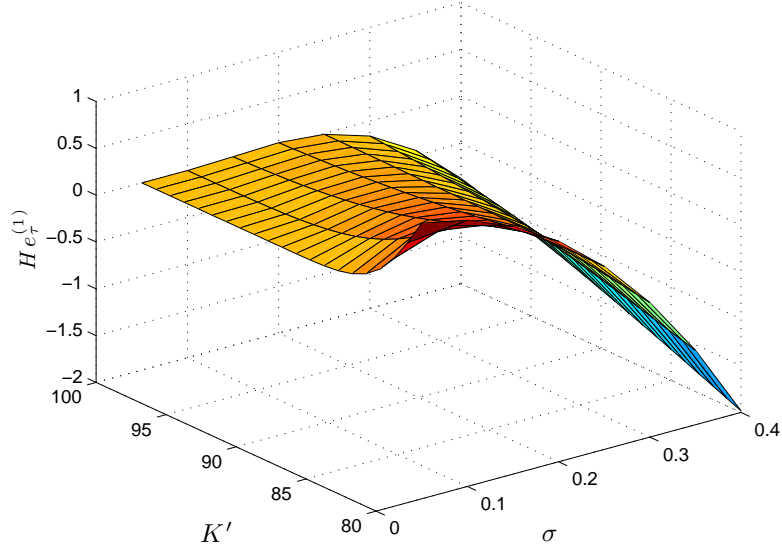}
		\caption{{\label{fig:fohe}{\bf{First order hedging error.}} The plot reports the first order hedging error  given in Equation (\ref{fohe}) as function of two variables: i) the strike of the hedging option, i.e. $K^{\prime} \in [80,100]$, ii) the diffusion coefficient of the underlying dynamics, i.e. $\sigma \in [0.05,0.4]$. The Figure is based on the following base case parameters' values: $K=80, r=0.03, T=1, \tau=0.6$.}}
		\end{figure}

Equation (\ref{fohe}) reports the first order hedging error in analytic closed form for the one dimensional case as function of all the model's parameters. Figure \ref{fig:fohe} reports the value of the first order hedging error by highlighting its dependency on two drivers: i) the strike of the hedging option $K^{\prime}$, ii) the diffusion coefficient $\sigma$ of the underlying dynamics. The plot highlights the influence of these two drivers on the magnitude of the first order hedging error that can be summarized as follows:
\begin{itemize}
\item when the strike of the hedging option $K^{\prime}$ reduces (becoming closer to $K$), the first order hedging error increases in absolute value; 
\item when the diffusion coefficient $\sigma$ of the underlying dynamics increases, the first order hedging error increases in absolute value. The intuition is that as $\sigma$ increases, a higher uncertainty characterizes the setting, turning out into an additional 'cost-component' associated to the timing risk.
\end{itemize}

\subsection{The second order hedge and its error}\label{sec:SO}

The existence of a second order semi-static hedge is stated and discussed in Section \ref{sec:HO}. Here we apply the general result to the one dimensional case in order to derive in analytic closed form the corresponding second order hedging error. A comparison between the first and second order hedging errors is also provided by showing the material error reduction obtained via the iteration procedure.\\

\begin{proposition}\label{p:SOHE}
Under Assumption \ref{oneDim}, the {\textit{second order hedging error}} at time $t$ associated to the stopping time $\tau$ is given in closed form as:

\begin{eqnarray}\label{eq:sohe}
\mathrm{He}^{(2)}_{\tau}&:=&\mu K 
e^{\frac{\sigma^2 (T-\tau)}{2}}\frac{1}{\sqrt{2 \pi \sigma^2 (s - \tau)}} 
\\&&\nonumber
 \times\int_\tau^T \, ds 
\left\{
\int_{\mathbf{R}} \mathrm{sgn} (u)
\left(
e^{\mu (s-\tau)}\mathcal{N}\left({d_3+\frac{u}{\sqrt{\sigma^2 (T-s)}}}
\right)
+
e^{-\mu (s-\tau)}\mathcal{N}\left({d_3-\frac{u}{\sqrt{\sigma^2 (T-s)}}}
\right)
\right)\, du
\right\}
\end{eqnarray}

with 
\begin{equation*}
{
d_3:=\frac{\log { \frac{K}{K'}} +\sigma^2(T-s)}{\sqrt{\sigma^2 (T-s)}} 
}
\end{equation*}
and $\mathcal{N}(\cdot)$ indicating the cumulative distribution function of a standard normal random variable.

\end{proposition}

\begin{proof}See Appendix \ref{pfSOHE}.
\qed
\end{proof}

Equation (\ref{eq:sohe}) reports in analytic closed form the second order hedging error in the one dimensional case expressed as function of all the model's parameters and represented via a simple mathematical structure. By leveraging on the proposed approach of re-iterating the timing risk identification and semi-static hedge decomposition in terms of knock-in options, higher order hedging errors with decreasing magnitude can be derived. Figure \ref{fig:sohe} depicts an example of the second order hedging error by highlighting its dependency on two drivers: i) the strike of the hedging option $K^{\prime}$, ii) the diffusion coefficient $\sigma$ of the underlying dynamics. The plot highlights the influence of these two drivers on the magnitude of the second order hedging error that can be summarized as follows:
\begin{itemize}
\item when the strike of the hedging option $K^{\prime}$ reduces (becoming closer to $K$), the second order hedging error increases in absolute value; 
\item when the diffusion coefficient $\sigma$ of the underlying dynamics increases, the second order hedging error increases in absolute value. 
\end{itemize}

\begin{figure}[ht!]
		\centering
		\includegraphics[scale=0.7]{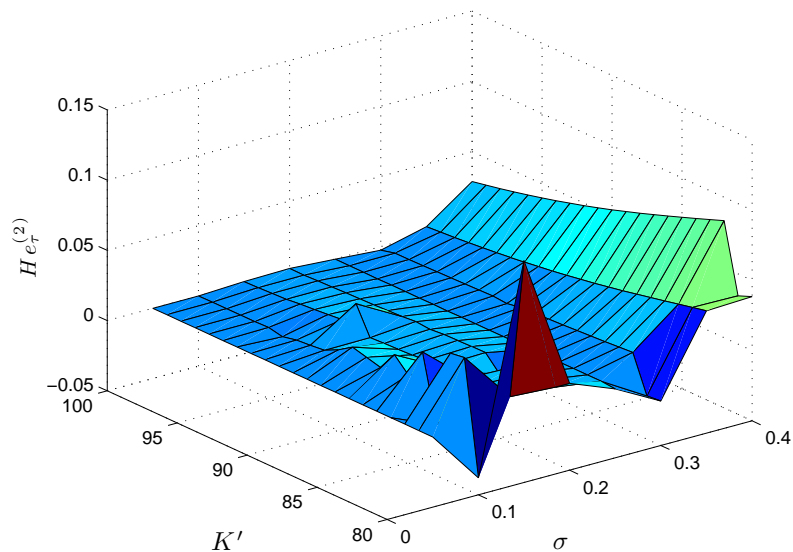}
		\caption{{\label{fig:sohe}{\bf{Second order hedging error.}} The plot reports the second order hedging error given in Equation (\ref{eq:sohe}) as function of two variables: i) the strike of the hedging option, i.e. $K^{\prime} \in [80,100]$, ii) the diffusion coefficient of the underlying dynamics, i.e. $\sigma \in [0.05,0.4]$. The Figure is based on the following base case parameters' values: $K=80, r=0.03, T=1, \tau=0.6$.}}
		\end{figure}
		
\begin{figure}[ht!]
		\centering
		\includegraphics[scale=0.7]{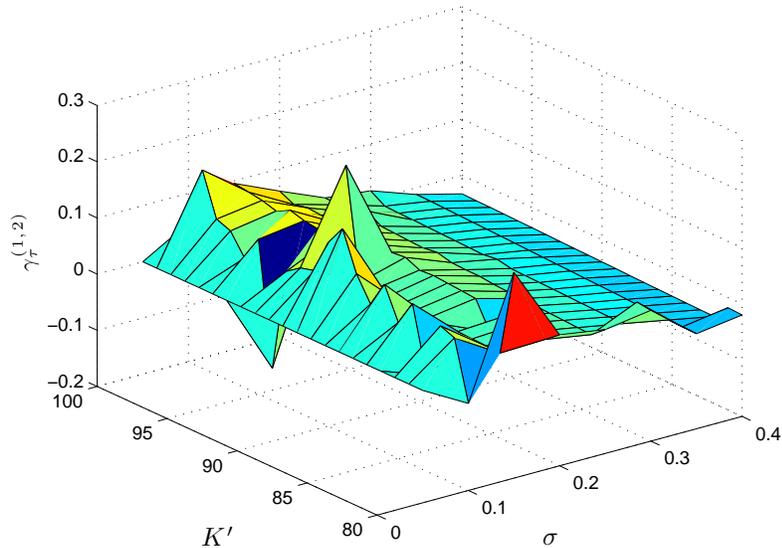}
		\caption{{\label{fig:rat}{\bf{Ratio between second order and first order hedging errors.}} The plot reports the ratio $\gamma^{(1,2)}_{\tau}$ between second order (Equation \eqref{eq:sohe}) and first order (Equation \eqref{fohe}) hedging errors given in Equation (\ref{rat_H2H1}) as function of two variables: i) the strike of the hedging option, i.e. $K^{\prime} \in [80,100]$, ii) the diffusion coefficient of the underlying dynamics, i.e. $\sigma \in [0.05,0.4]$. The Figure is based on the following base case parameters' values: $K=80, r=0.03, T=1, \tau=0.6$.}}
		\end{figure}
\subsection{First and second order hedging error comparison}\label{sec:FSOC}

Analyzing the behavior of both hedging errors (first and second order) allows to study their relative magnitude and thus the materiality of the hedging benefit gain by passing from order one to order two. Figure \ref{fig:rat} reports the ratio $\gamma^{(1,2)}_{\tau}$ between second order (Equation \eqref{eq:sohe}) and first order (Equation \eqref{fohe}) hedging errors defined as

\begin{equation}\label{rat_H2H1}
\gamma^{(1,2)}_{\tau}:=\frac{{He}^{(2)}_{\tau}}{{He}^{(1)}_{\tau}}
\end{equation}
 and its dependency on two main drivers: i) the strike of the hedging option $K^{\prime}$, ii) the diffusion coefficient $\sigma$ of the underlying dynamics. As we can see from the plot, the second order hedging error brings to a material reduction of the hedging 'cost' associated to the timing risk $\tau$, by dropping down the first order hedging error of $80-90\%$ in most cases, i.e. $1-|\gamma^{(1,2)}_{\tau}| \in [0.8,1]$. As we can see from the Figure, depending on the option specific characteristics, the hedging benefit gain can bring up to more than $90\%$ error reduction: this happens for example for high values of the diffusion coefficient, due to the different non-linear sensitivity of the first and second order hedging errors w.r.t. $\sigma$ parameter.

{

\section{Conclusions}\label{conc}
The paper addresses the problem of hedging positions on American-style options via static and semi-static hedging strategies in a multi-dimensional diffusion setting. 
In particular, the paper focuses on semi-static hedge of barrier options: this implies to deal with {\textit{timing risk}}, since the time at which the payoff will be paid by the option (even in the case of known amount) is {\textit{not known}} in advance. The starting point of our analysis is the work by \cite{CP}, where the authors show that the timing risk can be hedged via static positions in plain vanilla options. The present paper extends the static hedge formula proposed in \cite{CP} by giving sufficient conditions to decompose a {\textit{generalized timing risk}} into an integral of knock-in options in a multi-dimensional market model by considering a fairly large class of Markovian underlying dynamics. Then, a specific study of the semi-static hedge is conducted by defining the corresponding strategy based on positions in barrier options. A first order semi-static hedge is then built and discussed. The mathematical contribution of the paper then embeds a methodological proposal for the construction of higher order semi-static hedges, that can be interpreted as asymptotic expansions of the hedging error. This is done via an iterated procedure and the convergence of these higher order semi-static hedges to an exact hedge is shown. Finally, the paper provides an illustration of the main theoretical results for i) a symmetric case, ii) a one dimensional case. The first order and second order hedging errors are derived in analytic closed form for the one dimensional case and numerical results support the evidence about the material decrease of the hedging error by passing from the first to the second order ($80\%-90\%$ reduction in terms of error absolute value). 



\appendix
\section{Proofs}\label{tech_app}

This Appendix contains the proofs of the main theoretical results presented in the paper.

\subsection{Proof of Lemma \ref{fof2}}\label{prLemma}

Let us notice that pint 4. of Assumption \ref{aspGCP} allows us to write:
\begin{equation*}
\begin{split}
&\partial_s \{ \psi (s) q (s, x,z) p^F (t-s,z,y)\} \\
&= \psi'(s) q (s,x,z) p^F (t-s,z,y) + \psi (s) 
L^*_z q (s,x,z) p^F ({t-s},z, y) - \psi (s) q (s,x,z) 
\partial_s p^F ({t-s},z,y). 
\end{split}
\end{equation*}
Since $p^F(\cdot. z.y )$ is differentiable, we have: 
\begin{equation}\label{delta1}
\lim_{s \downarrow 0} 
\int_{\mathbf{R}^d} q (s,x,z) p^F ({t-s},z,y) \,dz = p^F (t, x, y).  
\end{equation}
Therefore, by using the result in (\ref{delta1}), Equation (\ref{delta2}) in point 2. of Assumption \ref{aspGCP} and the adjoint property of $ L^*_z $, 
we can write: 
\begin{equation*}\label{param2+}
\begin{split}
&\psi(t) q(t, x, y) - \psi (0) p^F(t, x, y)
\\
&= 
\lim_{\epsilon \downarrow o }\{
\psi(t- \epsilon)\int_{\mathbf{R}^d} 
q(t-\epsilon, x, z) p^F(\epsilon, z, y)dz
- \psi (\epsilon) \int_{\mathbf{R}^d}
q(\epsilon, x, z) p^F(t- \epsilon, z, y)dz\}
\\
& 
=\lim_{\epsilon \downarrow o }\{
\int_{\mathbf{R}^d}\int_{\epsilon}^{t-\epsilon} \partial_s\{
\psi(s)
q(s, x, z) p^F(t-s, z, y)\} ds dz\}
\\
& = \lim_{\epsilon \downarrow o }\{
\int_{\mathbf{R}^d}\int_{\epsilon}^{t-\epsilon} \psi'(s) q (s,x,z) p^F (t-s,z,y) + \psi (s) 
L^*_z q (s,x,z) p^F ({t-s},z, y) 
 \\& \qquad - \psi (s) q (s,x,z) 
\partial_s p^F ({t-s},z,y) ds dz\}
\\
& = \lim_{\epsilon \downarrow o }\{
\int_{\epsilon}^{t-\epsilon} \int_{\mathbf{R}^d}
q (s,x,z) h^{p^F} (t-s,z,y)  
ds dz\}
\\
& = 
\int_{0}^{t} \int_{\mathbf{R}^d}
q (s,x,z) h^{p^F} (t-s,z,y)  
ds dz.
\end{split}
\end{equation*}
The last equivalence holds due to the integrability of $h^{p^F}$. 
\qed


\subsection{Proof of Theorem \ref{GCP}}\label{pr_GCP}
Let us consider the l.h.s. of Equation (\ref{GTMR2}) and notice that the argument of the expectation can be written as:

\begin{equation}\label{errt11}
\begin{split}
& 1_{ \{ \tau \leq T \}} 
\psi(T-\tau) \mathbb{E}[ F(X_T) | \mathcal{F}_\tau]  
= 1_{ \{ \tau \leq T \}} 
\int_{\mathbf{R}^d} \psi(T-\tau) q (T-\tau, X_\tau, y) F(y) \,dy.
\end{split}
\end{equation}

Now, by leveraging on Equation (\ref{errt11}) and by applying Lemma \ref{fof2} to the r.h.s. of Equation (\ref{errt11}), we can write:
\begin{equation}\label{eq:cf}
\begin{split}
& 1_{ \{ \tau \leq T \}} 
\psi(T-\tau) \mathbb{E}[ F(X_T) | \mathcal{F}_\tau]  \\
&=
1_{ \{ \tau \leq T \}} 
\int_{\mathbf{R}^d} \left\{
 \psi(0) p^F(T- \tau,X_{\tau},y)
+\int_0^{T-\tau} \int_{\mathbf{R}^d} q (T-\tau - s, X_{\tau},z)
h^{p^F}(s,z,y)
 \,dz \, ds\right\}
F(y) \,dy
\\
&= 
1_{ \{ \tau \leq T \}} c_F 
+1_{ \{ \tau \leq T \}} 
\int_{\tau}^{T} \int_{\mathbf{R}^d} q (s-\tau , X_{\tau},z)
\left(\int_{\mathbf{R}^d}
h^{p^F}(T- s,z,y)F(y) \,dy\right)
 \,dz \, ds.
\end{split}
\end{equation}

By observing that the second term in the last row of Equation (\ref{eq:cf}) can be simplified as:
\begin{equation}
\begin{split}
&1_{ \{ \tau \leq T \}}
\int_{\tau}^{T} \int_{\mathbf{R}^d} q (s-\tau , X_{\tau},z)
\left(\int_{\mathbf{R}^d}
h^{p^F}(T-s,z,y)F(y) \,dy
\right) \, dz \, ds
\\&=1_{ \{ \tau \leq T \}}
\int_0^T 
1_{ \{ \tau \leq s \}}
\mathbb{E} [ \int_{\mathbf{R}^d}
h^{p^F}(T-s,X_s,y)F(y) \,dy
| \mathcal{F}_{\tau} ]
ds 
\\& = 
\int_0^T 
\mathbb{E} [ 1_{ \{ \tau \leq s,\ \tau \leq T  \} } \int_{\mathbf{R}^d}
h^{p^F}(T-s,X_s,y)F(y) \,dy
| \mathcal{F}_{\tau} ]
ds\\& = 
\int_0^T 
\mathbb{E} [ 1_{ \{ \tau \leq s\} } \int_{\mathbf{R}^d}
h^{p^F}(T-s,X_s,y)F(y) \,dy
| \mathcal{F}_{\tau} ]
ds,
\end{split}
\end{equation}
we get the desired result. 
\qed

\subsection{Proof of Theorem \ref{GSH}}
\label{pfGSH}
%
%

Let us recall that $X$ is a diffusion process and 
$\tau$ is the first exit time from a domain 
$ D \subset \mathbf{R}^d $; $ F $ is a measurable
function such that $ F = 0 $ on $ D^c $,
and $ \Psi \equiv 1 $. 
We associate a function $ \hat{F} $ such that $ \mathrm{supp} \, \hat{F} \subset D^c $, and write function $ \pi (F) = F - \hat{F} $. 

Thus, by applying Theorem \ref{GCP} to $ \pi (F) $, 
we have:
\begin{equation*}\label{GTMR3}
\begin{split}
\mathbb{E}[1_{ \{ \tau \leq T \} } \mathbb{E} [ \pi(F)(X_T)) |\mathcal{F}_\tau] | \mathcal{F}_t] 
= \int_0^T
\mathbb{E}[ 1_{ \{ \tau \leq s \}} 
\int_{\mathbf{R}^d}
h^{p^{\pi(F)}}(T-s, X_s, y)\pi(F) (y) dy
| \mathcal{F}_{\tau \wedge t} ] \,ds,
\end{split}
\end{equation*}
with $ h^{p^{\pi(F)}} $ defined in Equation \eqref{hpf} as 
\begin{equation*}
    h^{p^{\pi(F)}} (s,x,y) 
    := (L_x -\partial_s) p^{\pi(F)} (s, x,y).
\end{equation*}

Since we have $\mathbb{E}[ F (X_T) 1_{\{ \tau \geq T \} }|\mathcal{F}_t ]
=\mathbb{E} [ \pi(F)(X_T)1_{\{ \tau \geq T \} } |\mathcal{F}_t ] $, 
by observing that $ \hat{F}= 0 $ on $ D $, 
we obtain the results reported in Equations (\ref{KOGSH}) and (\ref{KIGSH}).
\qed

\subsection{Proof of Theorem \ref{HSS}}
\label{pfHSS}

The proof is done by induction. The case $ N = 1 $ follows by 
Theorem \ref{GSH}. Suppose that the result reported in Equation (\ref{HSSform}) is valid for $ N \geq 2 $. 
By applying \eqref{KIGSH} to $$ 1_{ \{ \tau \leq s \}}
 ( (\mathcal{S}^{p^{\pi}})^N_{T-s} (F) ) (X_s) 
\pi (F) $$ in place of $ 1_{ \{ \tau \leq s \}} F $, 
we have:
\begin{equation*}
\begin{split}
& \int_0^T
E [  1_{ \{ \tau \leq s \}} ( (\mathcal{S}^{p^{\pi}})^N_{T-s} 
(F) ) (X_s)| \mathcal{F}_{\tau \wedge t} ] \,ds \\
& = \int_0^T
\mathbb{E} [  \pi^\bot ( (\mathcal{S}^{p^{\pi}})^N_{T-s} 
(F) ) (X_s)| \mathcal{F}_{\tau \wedge t} ] \,ds \\
& \qquad 
+ \int_0^T \int_0^s \mathbb{E} [ 1_{ \{ \tau \leq u \}} 
(\mathcal{S}^{p^{\pi}})_{s-u}^1
( ( \mathcal{S}^{p^{\pi}})^N_{T-s} (F))  (X_u)| \mathcal{F}_{\tau \wedge t} ] \,duds.
\end{split}
\end{equation*}
Under Assumption \ref{aspHSS}, we
can change the order of the latter integral and thus obtain:
\begin{equation*}
\begin{split}
&\int_0^T \int_0^s \mathbb{E} [ 1_{ \{ \tau \leq u \}} (\mathcal{S}^{p^{\pi}})^1_{s-u}
( ( \mathcal{S}^{p^{\pi}})^N_{T-s} (F))  (X_u)| \mathcal{F}_{\tau \wedge t} ] \,duds \\
&= \int_0^T\mathbb{E}[1_{ \{ \tau \leq u \}} \int_u^T
(\mathcal{S}^{p^{\pi}})^1_{s-u} ( ( \mathcal{S}^{p^{\pi}})^N_{T-s} (F)) (X_u) \,ds| \mathcal{F}_{\tau \wedge t} ]du \\
&= \int_0^T \mathbb{E}[1_{ \{ \tau \leq u \}} \int_0^{T-u}
(\mathcal{S}^{p^{\pi}})^1_{s} 
( ( \mathcal{S}^{p^{\pi}})^N_{T-u-s} (F)) (X_u) \,ds| \mathcal{F}_{\tau \wedge t} ]du \\
& = \int_0^T \mathbb{E}[1_{ \{ \tau \leq u \}}  ( ( \mathcal{S}^{p^{\pi}})^{N+1}_{T-u} (F)) (X_u) \,ds| \mathcal{F}_{\tau \wedge t} ]du.
\end{split}
\end{equation*}
This proves the result stated in Equation (\ref{HSSform}). Under Assumption \ref{aspHSSc}, the same argument allows to prove the result stated in Equation (\ref{exact}) by induction.
\qed

\subsection{Proof of Proposition \ref{symm}}\label{pfsymm}
Under Assumption \ref{aspdiff}, it is sufficient to show that the following holds:
\begin{equation}\label{equvl1}
\begin{split}
\int_D q(t, X_{\tau}, y) F(y)dy 
&= \int_{D^c} q(t, X_{\tau}, y) F(\theta(y))dy.
\end{split}
\end{equation}

The equivalence in law of diffusion processes with  symmetric coefficients is stated under more general assumptions in \cite{AI}(ref. Lemma 3.3), including Equation \eqref{equvl1}. 
Let us recall that $ \theta $ is a reflection, then, by 
a change of variable $ \theta(y) = z $ 
we obtain the result stated in Equation \eqref{XPCS}. 
Moreover, by leveraging on Equation (\ref{adj0}), 
we can write
\begin{equation*}
\begin{split}
h^{q}(t,x,y) &= 
(L_x -\partial_t)q(t,x,y) = 0,
\end{split}
\end{equation*}
thus Assumption \ref{aspHSSc} is satisfied. 
\qed

\subsection{Proof of Theorem \ref{1dim}}\label{pf1dim}
Let us notice that the requirements stated at points 1. and 3. of Assumption \ref{aspHSS} are satisfied, 
since we are dealing with a Brownian motion with smooth drift; 
requirement stated at point 2. is the classical reflection principle. 
To show that also point 4. is satisfied, we estimate:
\begin{equation}\label{estimate-hpi}
\begin{split}
    |h^{p^\pi} (t,x,y)| 
&\leq C_b\frac{|x-y|}{t} \frac{1}{\sqrt{2\pi t} } e^{ 
    - \frac{(x-y)^2}{2t}  }
\\& = C_b
2^{\frac{3}{2}}t^{-\frac{1}{2}}\left\{\left(
\frac{(x-y)^2}{4t}  \right)^{\frac{1}{2}}
e^{ 
    - \frac{(x-y)^2}{4t}}
\right\}
\frac{1}{\sqrt{4\pi t} }e^{ 
    - \frac{(x-y)^2}{4t}}
\\& \leq C_b
2^{\frac{3}{2}}K_{\frac{1}{2}}t^{-\frac{1}{2}}
p^\pi(2t,x,y),
\end{split}
\end{equation}
where 
$C_b := ||\mu / \sigma ||_{\infty}  +  \frac{1}{2}||\sigma' ||_{\infty}$, 
and $K_{\frac{1}{2}}
:= || x^{\frac{1}{2}}e^{-x}||_{\infty}$. 

{Since the drift of $ Y $ is smooth and bounded, its transition density $ q $ satisfies (see e.g.\cite{MR0181836})}:
\begin{equation}\label{estimate-q}
\begin{split}
q(t, x,y)& \leq C_q t^{-\frac{1}{2}}
e^{-\frac{(x-y)^2}{4t}}
\\& = C_q 2 \pi^{\frac{1}{2}}
p^\pi(2t,x,y). 
\end{split}
\end{equation}
Therefore by (\ref{estimate-hpi}) and (\ref{estimate-q}), we have that 
\begin{equation*}
\begin{split}
&\int_0^T \int_{-\infty}^{\infty} 
\int_{-\infty}^{\infty}
|q(T-s, x,y) h^{p^\pi}(s,y,z)\pi(F_Y)(z)|dy dz ds
\\& \leq 
2^{\frac{5}{2}}\pi^{\frac{1}{2}}
C_q C_bK_{\frac{1}{2}}||F||_{\infty}
\int_0^T \int_{-\infty}^{\infty} 
\int_{-\infty}^{\infty}
s^{-\frac{1}{2}}
p^\pi(2(T-s),x,y)
p^\pi(2s,y,z)dy dz ds
\\& = 
2^{\frac{7}{2}}\pi^{\frac{1}{2}}
C_q C_bK_{\frac{1}{2}}||F||_{\infty}
T^{\frac{1}{2}}, 
\end{split}
\end{equation*}
that is, $q(T-s, x,y) h^{p^\pi}(s,y,z)\pi(F_Y)(z)$ 
is integrable in $[0,T] \times \mathbf{R}\times \mathbf{R}$. By using Equations (\ref{estimate-hpi}) and (\ref{estimate-q}), we then have:
\begin{equation*}
\begin{split}
&\int_0^{T} \int_0^{s_{N}}\int_0^{s_{N-1}} \cdots \int_0^{s_2} 
\int_{\mathbf{R}^{N}} q (T-s_N, x, y_N) |\pi(F_Y)(z)| 
\prod_{j=1}^{N} |h^{p^\pi} (s_{j} -s_{j-1}, y_{j}, y_{j-1})
d\mathbf{y} d s_1d s_2 \cdots d s_N  
\\& \leq 
 \int_0^{T} \int_0^{s_{N}}\int_0^{s_{N-1}} \cdots \int_0^{s_2} 
\int_{\mathbf{R}^{N}}
2^{\frac{3N}{2}+1}\pi^{\frac{1}{2}}
C_q C_b^NK_{\frac{1}{2}}^N||F||_{\infty}
\\& \qquad 
\times p^\pi(2(T-s_N),x,y)
\prod_{j=1}^{N}(s_j-s_{j-1})^{-\frac{1}{2}}
p^\pi(2(s_j-s_{j-1}),y_j,y_{j-1})
d\mathbf{y} d s_1d s_2 \cdots d s_N  
\\& = 
2^{\frac{3N}{2}+1}\pi^{\frac{1}{2}}
C_q C_b^NK_{\frac{1}{2}}^N||F||_{\infty}
\int_0^{T} \int_0^{s_{N}}\int_0^{s_{N-1}} \cdots \int_0^{s_2} 
\prod_{j=1}^{N}(s_j-s_{j-1})^{-\frac{1}{2}}
d s_1d s_2 \cdots d s_N  
\\
& = 
2^{\frac{3N}{2}+1}\pi^{\frac{1}{2}}
C_q C_b^NK_{\frac{1}{2}}^N||F||_{\infty}
B(\frac{1}{2}, \frac{1}{2}) 
B(1, \frac{1}{2}) 
B(\frac{3}{2}, \frac{1}{2}) 
\int_0^{T} \int_0^{s_{N}}\int_0^{s_{N-1}} \cdots \int_0^{s_5} 
s_4^{1} \prod_{j=5}^{N}(s_j-s_{j-1})^{-\frac{1}{2}}
d s_4 d s_5\cdots d s_N  
\\& = \cdots 
\\& = 
2^{\frac{3N}{2}+1}\pi^{\frac{1}{2}}
C_q C_b^NK_{\frac{1}{2}}^N||F||_{\infty}
\prod_{k=1}^{N-1} B(\frac{k}{2}, \frac{1}{2}) 
\int_0^{T}
s_N^{\frac{N-3}{2}} ds_N
\\& = 
2^{\frac{3N}{2}+2}\pi^{\frac{1}{2}}
C_q C_b^NK_{\frac{1}{2}}^N||F||_{\infty}
\prod_{k=1}^{N-1} B(\frac{k}{2}, \frac{1}{2}) 
(N-1)^{-1}
T^{\frac{N-1}{2}}. 
\end{split}
\end{equation*}
The last estimate also shows that Assumption \ref{aspHSSc} is satisfied, by proving the desired result.
\qed

\subsection{Proof of Proposition \ref{p:FOHE}.}\label{pfFOHE}
\begin{proof}By recalling Equation \eqref{he_diff}, the first order hedging error $\mathrm{He}^{(1)}_{\tau}$ at time $\tau$ can be written as:
{
\begin{equation}\label{he-dif-12}
\begin{split}
\mathrm{He}_{\tau}^{(1)}&:=
\mathbb{E}[( e^{X_{T -\tau }} -K')^+|X_0=\log K ]
-\mathbb{E}[( e^{2 \log K -X_{T -\tau }} -K')^{+}|X_{0}=\log K  ]
\\& = I - II
\end{split}
\end{equation}
on the set $\{\tau < T\}$. }
\if1 
As a consequence, we need to characterize both terms of the r.h.s. of Equation \eqref{he_diff}. We can write
\begin{equation}\label{he_a}
\mathbb{E}[( e^{X_{T-\tau}} -K')^+|X_0=\log K]
= \mathbb{E}[( e^{X_{T-t}} -K')^+|X_0=K] \bigg|_{t= \tau},
\end{equation}
and 
\begin{equation}\label{he_b}
\mathbb{E}[( e^{2 \log K -X_{T-\tau}} -K')|X_0=\log K, \tau]
= \mathbb{E}[( e^{2 \log K -X_{T-t}} -K')|X_0= \log K] \bigg|_{t= \tau}.
\end{equation}
\fi 
For the first term of \eqref{he-dif-12}, we have that 
\begin{equation}\label{ft}
\begin{split}
I & = \mathbb{E}[( e^{X_{T -\tau }} -K')^+|X_0=\log K]\\
&=
\int_{\log K'}^{\infty}(e^{y}-K')^+ 
\left\{ \frac{1}{\sqrt{2\pi\sigma^2(T -\tau )}}
e^{-\frac{(y-(\log K+\mu (T -\tau )))^2}{2\sigma^2(T -\tau )}}\right\}dy\\
&=Ke^{(\frac{\sigma^2}{2} +\mu) (T -\tau ) }
\mathcal{N}\left(\frac{\log {\frac{K}{K'}}{+} (\mu + \sigma^2)(T -\tau )}{\sqrt{\sigma^2 (T -\tau )}}\right)-K' \mathcal{N}\left(\frac{\log {\frac{K}{K'}} {+} \mu (T -\tau )}{\sqrt{\sigma^2 (T -\tau )}}\right)
\end{split}
\end{equation}

In a similar way, the second term of \eqref{he-dif-12} can be written as:
\begin{equation}\label{st}
\begin{split}
II &=\mathbb{E}[( e^{2 \log K -X_{t}} -K')^+|X_0= \log K]\\
&=
\int_{-\infty}^{2\log K-\log K'}(e^{2 \log K-y}-K')^+ 
\left\{ \frac{1}{\sqrt{2\pi\sigma^2(T -\tau )}}
e^{-\frac{(y-(\log K+\mu (T -\tau )))^2}{2\sigma^2(T -\tau )}}\right\}dy\\
&=
Ke^{(\frac{\sigma^2 }{2} - \mu )(T -\tau )}
\mathcal{N} 
\left(\frac{\log {\frac{K}{K'}}{-} (\mu -\sigma^2)(T -\tau )}{\sqrt{\sigma^2 (T -\tau )}}\right)
-K' \mathcal{N}\left(\frac{\log {\frac{K}{K'}}{-}\mu (T -\tau )}{\sqrt{\sigma^2 (T -\tau )}}\right).
\end{split}
\end{equation}
Substituting Equation \eqref{ft} and \eqref{st} into Equation \eqref{he_diff} leads to the desired result.
\qed

\end{proof}

\subsection{Proof of Proposition \ref{p:SOHE}.}\label{pfSOHE}
\begin{proof}
The second order hedge is given by integrating w.r.t. variable $s$ the options 
with payoff function
$$
\int_{\mathbf{R}} \pi^{\bot}
h^{p^{\pi}}(T-s,X_s,y)\pi (F)(y) \,dy,$$
{and the hedging error (see \eqref{eq:cost_in}) 
is given by 
$$
\int_{\mathbf{R}} {\pi}
h^{p^{\pi}}(T-s,X_s,y)\pi (F)(y) \,dy,$$
}
that can be equivalently written under the following form:

\begin{equation}\label{sp_so}
\begin{split}
   & \int_{\mathbf{R}} \pi
h^{p^{\pi}}(T-s,X_s,y)\pi (F)(y) \,dy  \\
&=   \int_{\mathbf{R}} 
h^{p^{\pi}}(T-s,X_s,y) (e^{y} -K')^+ 1_{\{ X_s { \geq} \log K \} } \,dy  \\
&  -  \int_{\mathbf{R}} 
h^{p^{\pi}}(T-s,X_s,y) (e^{2\log K - y} -K')^+ 
1_{\{ X_s { \geq}  \log K \} } \,dy  \\
& {-}   \int_{\mathbf{R}} 
h^{p^{\pi}}(T-s,2 \log K - X_s,y)(e^{y} -K')^+1_{\{ X_s { \leq}  \log K \} }  \,dy \\
& {+}   \int_{\mathbf{R}} 
h^{p^{\pi}}(T-s,2 \log K - X_s,y)(e^{2\log K - y} -K')^+1_{\{ X_s { \leq}\log K \} }  \,dy,
\end{split}
\end{equation}
with $h^{p^{\pi}}$  defined in Equation \eqref{hppi0} as
\begin{equation*}
\begin{split}
h^{p^{\pi}}(t,x,y)&=
	\mu \frac{\partial}{\partial x} 
	\frac{1}{\sqrt{2\pi  \sigma^2 t}} 
	e^{ - \frac{(x-y)^2}{2\sigma^2 t} }
\\    & = - \mu \frac{\partial}{\partial y} 
    	\frac{1}{\sqrt{2\pi  \sigma^2 t}}  e^{ 
    - \frac{(x-y)^2}{2\sigma^2 t} }
. 
\end{split}
\end{equation*}

Starting from Equation \eqref{sp_so}, the second order hedging error at $ \tau $, namely $ \mathrm{He}^{(2)}_{\tau}$, 
can be obtained by introducing the expectation operator inside the integrals, as:
\begin{equation}\label{sohe_int}
\begin{split}
\mathrm{He}^{(2)}_{\tau} &:= \int_\tau^T E[ \int_{\mathbf{R}} 
h^{p^{\pi}}(T-s,X_s,y) (e^{y} -K')^+ 1_{\{ X_s {\geq} \log K \}} \,dy 
|\mathcal{F}_\tau ] \,ds  \\
& - \int_\tau^T  E[ \int_{\mathbf{R}} 
h^{p^{\pi}}(T-s,X_s,y) (e^{2\log K - y} -K')^+ 
1_{\{ X_s {\geq} \log K \} } \,dy |\mathcal{F}_\tau ]\, ds \\
& {-} \int_\tau^T E[ \int_{\mathbf{R}} 
h^{p^{\pi}}(T-s,2 \log K - X_s,y)(e^{y} -K')^+ 
1_{\{ X_s { \leq} \log K \}} \,dy 
|\mathcal{F}_\tau ]\, ds \\
& {+} \int_\tau^T 
E[ \int_{\mathbf{R}} 
h^{p^{\pi}}(T-s,2 \log K - X_s,y)(e^{2\log K - y} -K')^+ 
1_{\{ X_s {\leq} \log K \}} \,dy
|\mathcal{F}_\tau ]\, ds.\\
\end{split}
\end{equation}
Each one of the four terms appearing in the r.h.s. of equation \eqref{sohe_int} characterizing the second order hedging error can be simplified via algebraic calculations as follows.\\

Let us consider the first term in the r.h.s. of Equation \eqref{sohe_int}. This term can be written in an equivalent form as:
\begin{equation}\label{1_split}
\begin{split}
& \int_\tau^T E[ \int_{\mathbf{R}} 
h^{p^{\pi}}(T-s,X_s,y) (e^{y} -K')^+ 1_{\{ X_s { \geq} \log K \}} \,dy 
|\mathcal{F}_\tau ] \,ds  \\
& = \int_\tau^T E[ \int_{\mathbf{R}} 
- \mu\left(\frac{\partial}{\partial y} 
    	\frac{1}{\sqrt{2\pi  \sigma^2 (T-s)}}  e^{ 
    - \frac{(X_s-y)^2}{2\sigma^2 (T-s)} }\right)
(e^{y} -K')^+ 1_{\{ X_s {\geq}  \log K \}} \,dy 
|\mathcal{F}_\tau ] \,ds 
\\
& = \int_\tau^T E[ \int_{\mathbf{R}} 
\mu \frac{1}{\sqrt{2\pi  \sigma^2 (T-s)} } 
e^{ - \frac{(y-X_s - \sigma^2 (T-s) )^2}{2\sigma^2 (T-s)} }
e^{ \frac{\sigma^2 (T-s)}{2} + X_s }
1_{\{y > \log K'\}}
1_{\{ X_s {\geq}  \log K \}} \,dy 
|\mathcal{F}_\tau ] \,ds 
\\
& = \mu\int_\tau^T 
e^{ \log K + \mu (s - \tau ) + \frac{\sigma^2 (T-\tau )}{2}}
{\int_{0}^{\infty}}
 \mathcal{N}\left( 
\frac{\log {\frac{K}{K'} +}u {+}\sigma^2 (T-s)}{\sqrt{\sigma^2 (T-s)}}\right)
\frac{1}{\sqrt{2 \pi \sigma^2 (s - \tau)}}
e^{-\frac{(u - (\mu + \sigma^2) (s-\tau ))^2}
{2 \sigma^2 (s-\tau ) } }
du ds, \\
& = \mu K\int_\tau^T 
e^{ \mu (s - \tau ) + \frac{\sigma^2 (T-\tau )}{2}}
{\int_{0}^{\infty}}
\mathcal{N}\left(
\frac{\log {\frac{K}{K'} +}u {+}\sigma^2 (T-s)}{\sqrt{\sigma^2 (T-s)}}\right)
\frac{1}{\sqrt{2 \pi \sigma^2 (s - \tau)}}
e^{-\frac{(u - (\mu + \sigma^2) (s-\tau ))^2}
{2 \sigma^2 (s-\tau ) } }
du ds,
\end{split}
\end{equation}


Let us consider the second term in the r.h.s. of Equation \eqref{sohe_int}. This term can be written in an equivalent form as:
\begin{equation*}
\begin{split}
&\int_\tau^T  E[ \int_{\mathbf{R}} 
h^{p^{\pi}}(T-s,X_s,y) (e^{2\log K - y} -K')^+ 
1_{\{ X_s {\geq}  \log K \} } \,dy  ds|\mathcal{F}_\tau ]\, ds\\
& = {-}\int_\tau^T E[ \int_{\mathbf{R}} 
\mu \frac{1}{\sqrt{2\pi  \sigma^2 (T-s)}}  e^{ 
    - \frac{(X_s-y)^2}{2\sigma^2 (T-s)} }
e^{2 \log K -y} 1_{\{ y < {2 \log K - \log K'} \}}1_{\{ X_s {\geq}  \log K \}} \,dy 
|\mathcal{F}_\tau ] \,ds 
\\
& = - \mu {K}\int_\tau^T 
e^{- \mu(s-\tau )+ \frac{\sigma^2 (T-\tau )}{2} }
{\int_{0}^{\infty}}
\mathcal{N}\left( {
 \frac{\log {\frac{K}{K'}}  - u + \sigma^2(T-s)}{\sqrt{\sigma^2 (T-s)}}
 }
\right) 
\frac{1}{\sqrt{2 \pi \sigma^2 (s - \tau)}}
e^{-\frac{(u - (\mu- \sigma^2) (s-\tau ))^2}{2 \sigma^2 (s-\tau ) }- u}
\,du \,ds.
\end{split}
\end{equation*}

Let us consider the third term in the r.h.s. of Equation \eqref{sohe_int}. This term can be written in an equivalent form as:

\begin{equation}
\begin{split}
&\int_\tau^T E[ \int_{\mathbf{R}} 
h^{p^{\pi}}(T-s,2 \log K - X_s,y)(e^{y} -K')^+ 1_{\{ X_s {\leq } \log K \}} \,dy 
|\mathcal{F}_\tau ]\, ds \\
 & = 
\int_\tau^T E[ \int_{\mathbf{R}} 
    \mu
    	\frac{1}{\sqrt{2\pi  \sigma^2 (T-s)} } 
e^{  - \frac{(2 \log K - X_s-y)^2}{2\sigma^2 (T-s)}}
    e^{y}1_{\{y > \log K'\}}1_{\{ X_s  { \leq } \log K \}} \,dy 
|\mathcal{F}_\tau ] \,ds 
\\
& = 
    \mu
\int_\tau^T E[ 
e^{ {-X_s +2 \log K} +  \frac{\sigma^2 (T-s)}{2}}
\int_{\mathbf{R}} 
    	\frac{1}{\sqrt{2\pi  \sigma^2 (T-s)} } 
e^{  - \frac{(y- 2 \log K + X_s-\sigma^2 (T-s))^2}{2\sigma^2 (T-s)}}
    1_{\{y > \log K'\}}1_{\{ X_s {\leq } \log K \}} \,dy 
|\mathcal{F}_\tau ] \,ds 
\\
& = 
\mu {
K}
\int_\tau^T 
{\int^0_{-\infty}}
e^{ {
- \mu (s- \tau ) }+ \frac{\sigma^2 (T-\tau)}{2}}
\mathcal{N}\left(
{
\frac{\log  {\frac{K}{K'}}  {-} u {+}\sigma^2 (T-s)}{\sqrt{\sigma^2 (T-s)}}
}
\right)
\frac{1}{\sqrt{2 \pi \sigma^2 (s - \tau)}}
e^{-\frac{(u -({{\mu - \sigma^2}})(s-\tau ))^2}{2 \sigma^2 (s-\tau )}}
\,du\,ds.\\
\end{split}
\end{equation}

Let us consider the fourth term in the r.h.s. of Equation \eqref{sohe_int}. This term can be written in an equivalent form as:

\begin{equation}\label{f_split}
\begin{split}
&\int_\tau^T 
E[ \int_{\mathbf{R}} 
h^{p^{\pi}}(T-s,2 \log K - X_s,y)(e^{2\log K - y} -K')^+ 1_{\{ X_s {\leq }  \log K \}} \,dy|\mathcal{F}_\tau ]\, ds\\
\\& = {-} \mu 
\int_\tau^T E[ 
e^{X_s + \frac{\sigma^2 (T-s)}{2}}\int_{\mathbf{R}} \frac{1}{\sqrt{2\pi  \sigma^2 (T-s)}}
e^{ - \frac{(y - 2\log K +X_s+ \sigma^2(T-s))^2}{2\sigma^2 (T-s)} }
1_{\{ y < {2\log K - \log K'} \}}1_{\{ X_s {\leq }  \log K \}} \,dy |\mathcal{F}_\tau ] \,ds 
\\& = {-} \mu {K}\int_\tau^T
{ \int^0_{-\infty}}
e^{\mu (s-\tau) + \frac{\sigma^2 (T-\tau)}{2}}
\mathcal{N}\left( {
\frac{\log {\frac{K}{K'}} + u +\sigma^2(T-s)}{\sqrt{\sigma^2 (T-s)}} 
}
\right)
\frac{1}{\sqrt{2 \pi \sigma^2 (s - \tau)}}
e^{-\frac{(u -(\mu+ \sigma^2) (s-\tau ))^2}{2 \sigma^2 (s-\tau )}}
\,du \,ds. 
\end{split}
\end{equation}

Finally, by substituting Equations \eqref{1_split}-\eqref{f_split} into the r.h.s. of Equation \eqref{sohe_int}, we can write the second order hedging error at $ \tau $ as

\begin{equation}
\begin{split}
   \mathrm{He}^{(2)}_{\tau}&=\\
  &  =\mu K\int_\tau^T 
{\int_{0}^{\infty}}
e^{ \mu (s - \tau ) + \frac{\sigma^2 (T-\tau )}{2}}
\mathcal{N}\left(
\frac{\log {\frac{K}{K'} +}u {+}\sigma^2 (T-s)}{\sqrt{\sigma^2 (T-s)}}\right)
\frac{1}{\sqrt{2 \pi \sigma^2 (s - \tau)}}
e^{-\frac{(u - (\mu + \sigma^2) (s-\tau ))^2}
{2 \sigma^2 (s-\tau ) } }
\, du\, ds
\\&   {+} \mu {K}\int_\tau^T 
{\int_{0}^{\infty}}
e^{- \mu(s-\tau )+ \frac{\sigma^2 (T-\tau )}{2} }
\mathcal{N}\left( {
 \frac{\log { \frac{K}{K'}}  - u + \sigma^2(T-s)}{\sqrt{\sigma^2 (T-s)}}
 }
\right) 
\frac{1}{\sqrt{2 \pi \sigma^2 (s - \tau)}}
e^{-\frac{(u - (\mu- \sigma^2) (s-\tau ))^2}{2 \sigma^2 (s-\tau ) }}
\,du \,ds
\\&  -  \mu {K}
\int_\tau^T 
{\int^0_{-\infty}}
e^{ {
- \mu (s- \tau ) }+ \frac{\sigma^2 (T-\tau)}{2}}
\mathcal{N}\left(
{
\frac{\log  { \frac{K}{K'}}  {-} u {+}\sigma^2 (T-s)}{\sqrt{\sigma^2 (T-s)}}
}
\right)
\frac{1}{\sqrt{2 \pi \sigma^2 (s - \tau)}}
e^{-\frac{(u -({{\mu - \sigma^2}})(s-\tau ))^2}{2 \sigma^2 (s-\tau )}}
\,du\,ds
\\ &   - \mu {K}\int_\tau^T
{ \int^0_{-\infty}}
e^{\mu (s-\tau) + \frac{\sigma^2 (T-\tau)}{2}}
\mathcal{N}\left( {
\frac{\log { \frac{K}{K'}} + u +\sigma^2(T-s)}{\sqrt{\sigma^2 (T-s)}} 
}
\right)
\frac{1}{\sqrt{2 \pi \sigma^2 (s - \tau)}}
e^{-\frac{(u -(\mu+ \sigma^2) (s-\tau ))^2}{2 \sigma^2 (s-\tau )}}
\,du \,ds,
\\& = 
\mu K 
e^{\frac{\sigma^2 (T-\tau)}{2}}\frac{1}{\sqrt{2 \pi \sigma^2 (s - \tau)}}
\int_\tau^T \, ds 
\\& \qquad \times 
\left\{
\int_{\mathbf{R}} \mathrm{sgn} (u)
\left(
e^{\mu (s-\tau)}\mathcal{N}\left( {
\frac{\log { \frac{K}{K'}}+\sigma^2(T-s)+u}{\sqrt{\sigma^2 (T-s)}} 
}
\right)
+
e^{-\mu (s-\tau)}\mathcal{N}\left( {
\frac{\log { \frac{K}{K'}} +\sigma^2(T-s)-u}{\sqrt{\sigma^2 (T-s)}} 
}
\right)
\right)\, du
\right\}
\end{split}
\end{equation}
and obtain the desired result.
\qed
\end{proof}
}


\begin{thebibliography}{99}
\bibitem{FJY2} 
Akahori, J. Barsotti, F., and Imamura, Y. 
``Asymptotic Static Hedge via Symmetrization",
working paper. 

\bibitem{AI}
Akahori, J. and Imamura, Y. 
``On a symmetrization of diffusion processes'', 
Quantitative Finance 14 (2014), no. 7, 1211--1216. 



\bibitem{BK}
Bally, K. and Kohatsu-Higa, A. (2015)
``A probabilistic interpretation of the parametrix method'',
Annals of Applied Probability, 
Volume 25, Number 6, 3095-3138. 

\bibitem{BC}
Bowie, J. and Carr, P 
(1994)
``Static Simplicity'', Risk, 7(8), , 44--50.

\bibitem{CC} Carr, P. and Chou A. (1996), ``Breaking Barriers'', Risk, 10(9), 139--145

\bibitem{CEG} Carr, P., Ellis K., Gupta V. (1998), ``Static Hedging of Exotic Options'', Journal of Finance, 1165--1190

\bibitem{CL} Carr, P., and Lee R. (2009), ``Put-Call Symmetry: Extensions and Applications'', Mathematical Finance, 19(4), 523--560

\bibitem{CN} Carr, P. and Nadtochiy, S. (2011) 
``Static Hedging under Time-homogeneous Diffusions", 
SIAM Journal on Financial Mathematics, 2(1),  794--838

\bibitem{CP} Carr P., and Picron, J. 
(1999), ``Static Hedging of Timing Risk'', 
Journal of Derivatives 6, 57--70.

\bibitem{C}
Corielli, F, Fosci, P. and Pascucci, A. (2010) 
``Parametrix approximation of diffusion transition densities'', SIAM Journal on Financial Mathematics, 1, 837--867.






\bibitem{DEK}
Derman, E., Ergener, D., and Kani, I. (1995). ``Static options replication. Journal of Derivatives", 2:78-95

\bibitem{FINK}
Fink, J. (2003)
``An examination of the effectiveness of static hedging in the presence of stochastic volatility", 
Journal of Futures Markets, 23(9):859-890.

\bibitem{MR0181836} Friedman, A. 
{\it Partial differential equations of parabolic type}, 1964 by Prentice-Hall,
2008 by Dover. 

\bibitem{HL} 
Henry-Labord\`ere, P. 
{\it A General Asymptotic Implied Volatility
for Stochastic Volatility Models}, arxiv.org/pdf/cond-mat/0504317

\bibitem{II}
Ida, Y., and Imamura, Y.(2016)
``Towards the Exact Simulation Using Hyperbolic Brownian Motion'', working paper.  

\bibitem{Yur} Imamura, Y. (2011) 
``A remark on static hedging of options written on the last exit time'', Review of Derivatives Research, 14, pp. 333-347

\bibitem{ImTa}
Imamura, Y., and Takagi, K.
(2013), ``
Semi-Static Hedging Based on a Generalized Reflection Principle on a Multi Dimensional Brownian Motion", 
Asia-Pacific Financial Markets,
(20) 71-81.

\bibitem{KaTaYaJSIAM}
Kato, T., Takahashi, A., and Yamada, T.(2013)
``An Asymptotic Expansion Formula for Up-and-Out Barrier Option Price under Stochastic Volatility Model," 
JSIAM Letters, Vol. 5,
pp.17-20.

\bibitem{KaTaYa}
Kato, T., Takahashi, A., and Yamada, T.(2014)
``A Semi-group Expansion for Pricing Barrier Options''
International Journal of Stochastic Analysis, Volume 2014, ArticleID 268086, 15pages

\bibitem{EEL}
Levi, E. E. (1907) ``Sulle equazioni lineari totalmente ellittiche alle derivate parziali'', Rendiconti del Circolo Matematico di Palermo (in Italian) 24 (1): 275--317

\bibitem{NP}
Nalholm, M., and 
Poulsen, R. (2006)
``Static Hedging of Barrier Options under General
Asset Dynamics: Unification and Application,"  Journal of Derivatives, 13 (4), 
pp. 46-60.

\bibitem{ShTaTo}
Shiraya, K., Takahashi, A.
and Toda, M. (2011)
``Pricing Barrier and Average
Options under Stochastic Volatility Environment",  
Journal of Computational
Finance, 15.2 (Winter 2011/2012): 111-148.

\bibitem{ShTaYa}
Shiraya, K., Takahashi,
and Yamada, T. (2012) 
"Pricing Discrete Barrier Options Under Stochastic Volatility"
Asia-Pacific Financial Markets, Vol. 19-3, pp 205-232. 

\bibitem{Ta}
Takahashi, A.(1999) 
``An Asymptotic Expansion Approach to Financial Contingent Claims,"
Asia-Pacific Financial Markets, Vol. 6, pp.115-151. 



\end{thebibliography}
\end{document}